\documentclass[a4paper]{article}
\usepackage{sections/preamble}
\title{Chiral life on a slab }

\author[1]{Sergey Alekseev}
\author[3]{Mykola Dedushenko}
\author[1,2]{Mikhail Litvinov}
\affil[1]{Department of Physics and Astronomy, Stony Brook University, Stony Brook, NY 11794-3800, USA}
\affil[2]{C. N. Yang Institute for Theoretical Physics, Stony Brook University, Stony Brook, NY 11794-3800, USA}
\affil[3]{Simons Center for Geometry and Physics,
Stony Brook University, Stony Brook, NY 11794-3636, USA}

\begin{document}

\maketitle

\begin{abstract}
We study chiral algebra in the reduction of 3D $\mathcal{N} = 2 $ supersymmetric gauge theories on an interval with the $\cN=(0,2)$ Dirichlet boundary conditions on both ends.
By invoking the 3D ``twisted formalism'' and the 2D $\beta\gamma$-description we explicitly find the perturbative $\overline{Q}_+$ cohomology of the reduced theory.
It is shown that the vertex algebras of boundary operators are enhanced by the line operators.
A full non-perturbative result is found in the abelian case, where the chiral algebra is given by the rank two Narain lattice VOA, and two more equivalent descriptions are provided. Conjectures and speculations on the nonperturbative answer in the non-abelian case are also given.

\end{abstract}
\newpage

\tableofcontents
\section{Introduction}

The role of vertex operator algebras (VOA) in theoretical physics and mathematics has vastly expanded over the recent decades, way beyond their original scope \cite{Belavin:1984vu,Borcherds:1983sq,FreLepMeur88}.
 This includes, among others, applications of VOAs in differential topology of four- and three-manifolds \cite{Dedushenko:2017tdw,Feigin:2018bkf,Cheng:2022rqr}, and in higher-dimensional QFT  \cite{Beem:2013sza,Beem:2014kka,Costello:2018fnz,Costello:2018swh}.
 While some of these topics are more developed \cite{Beem:2013sza}, the existing literature only scratches the surface of the topological applications and some other topics, such as boundary algebras in \cite{costello2020boundary}.
 Our main interest in the current paper is an interplay between the older role of VOAs as chiral algebras in 2D $\cN=(0,2)$ theories \cite{Witten:1993jg,Silverstein:1994ih}, and their modern  appearance as boundary algebras supported by the $(0,2)$ boundary conditions in 3D $\cN=2$ theories \cite{costello2020boundary}.
 As will be explained later, this is also motivated by applications to the four-manifolds, following \cite{Dedushenko:2017tdw,Feigin:2018bkf}, see also earlier works \cite{Gadde:2013sca,Assel:2016lad}.

The main subject of this work is a 3D $\cN=2$ theory placed on an interval with $\cN=(0,2)$ boundary conditions on both ends.
This basic setup is also explored in a companion paper \cite{DL1} from the more physical perspective, where we compute the effective 2D $\cN=(0,2)$ action in the infrared (IR) limit of the interval-reduced 3D model.
It is known that 2D $\cN=(0,2)$ theories contain VOAs as their chiral algebras in the $\bar{Q}_+$-cohomology of local observables \cite{Witten:1993jg}.
Being invariant under the renormalization group (RG) flow \cite{DedLGChiral}, such a chiral algebra must admit a UV realization in a 3D theory on the interval.
It consists of a pair of (possibly different) VOAs, associated to the boundaries, extended by a category of their bi-modules associated to the appropriate line operators stretched between the boundaries (see an illustration on Fig.
 \ref{fig:boundary-conditions}).

\begin{figure}[h]
	\begin{center}
    \def\svgwidth{0.6\columnwidth}
    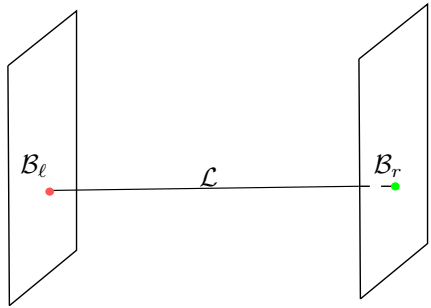

	\caption{Two boundaries with a line operator connecting them.}
	\label{fig:boundary-conditions}
	\end{center}
\end{figure} 
Generally speaking, one expects various types of line operators to appear, including the descendant lines compatible with our supercharge \cite{costello2020boundary}.
In the simplest case of pure $\cN=2$ super Yang-Mills (SYM) with Chern-Simons (CS) level, which will be our main focus in the bulk of this paper, we will fully describe the spectrum of relevant lines. In particular, Wilson lines (which are descendants of the ghost field) play an important part in this story.

This connects us to another topic that has seen a lot of interest recently, -- line defects and their role in related or similar constructions.
For example, holomorphic boundary conditions in 3D topological QFT (TQFT) may support non-trivial (relative) rational CFT (RCFT), and such TQFT/RCFT correspondence \cite{Witten:1988hf,Frohlich:1989gq} has been studied in great detail \cite{Felder:1999cv,Fuchs:2001am,Fuchs:2002cm,Fuchs:2003id,Fuchs:2004dz,Fr_hlich_2004}.
Bulk topological lines that are parallel to the boundary lead in this case to the boundary topological lines, while the bulk topological lines piercing the boundary give modules of the boundary chiral algebra.
Fusion of such lines corresponds to fusion of the VOA modules.
Note that in this story one also naturally encounters interval reductions: A TQFT on an interval (with the appropriate boundary conditions) leads to the full (non-chiral) RCFT, and segments of line operators stretched between the boundaries generate its primary fields.
A more sophisticated class of examples comes from the topologically twisted 3D $\cN=4$ theories, whose categories of line defects were recently studied in \cite{Dimofte_2020}, see also \cite{Creutzig:2021ext,Garner:2022vds,Garner:2022rwe}.
Such theories may also possess holomorphic boundary conditions \cite{Costello:2018fnz} supporting certain boundary VOAs, and the bulk lines piercing boundaries generate their modules as well.
This leads to interesting questions of determining categories of VOA modules corresponding to the bulk lines \cite{Wenjun_bega}, and it has also been argued that moduli spaces of vacua of the underlying physical theory can be recovered from the knowledge of boundary VOAs and their categories of modules \cite{Costello:2018swh}.

We are interested in similar constructions applied to three-dimensional theories with $\cN=2$ supersymmetry (including 3D $\cN=4$ viewed as $\cN=2$). The $\cN=(0,2)$ boundary conditions \cite{Gadde:2013wq,Okazaki:2013kaa,Yoshida:2014ssa} in such theories support non-trivial boundary VOAs in the $\bar{Q}_+$-cohomology, which is a direct analog of the 2D chiral algebra construction.
The $\cN=(0,2)$ boundary conditions are compatible with the holomorphic-topological (HT) twist in 3D $\cN=2$ \cite{Aganagic:2017tvx,costello2020boundary}. Thus, one may focus entirely on such HT-twisted theories, which often simplifies the VOA-related questions. The HT twist is also compatible with complexified Wilson lines, vortex lines in abelian \cite{Drukker_Vortex} and non-abelian \cite{Hosomichi:2021gxe} cases and their generalizations. Unlike in the topologically twisted theories, these line operators are not fully topological and cannot have arbitrary shape. They are supported along a straight line in the topological direction and are point-like in the holomorphic plane.

In general, one could choose different left and right boundary conditions $\cB_\ell$ and $\cB_r$, leading to the left and right boundary chiral algebras $V_\ell$ and $V_r$ and relations between them.
For example, often $\cB_\ell=\cB$ admits a counterpart $\cB_r=\cB^\perp$, such that the interval theory is trivially gapped in the IR, with the trivial chiral algebra. 
This would imply an interesting ``duality'' between the boundary chiral algebras, perhaps of relevance to some of the questions studied in \cite{Cheng:2022rqr}.
The only pairs of $\cN=(0,2)$ boundary conditions on the interval that appear in the literature so far include: Neumann-Neumann for gauge theories in \cite{Sugiyama:2020uqh}, and Dirichlet-Neumann in \cite{Dedushenko:2021mds,Bullimore:2021rnr}.

In this paper we will focus on the simplest nontrivial case of 3D $\cN=2$ pure SYM with group $G$ and CS level $k$, which is placed on the interval with $\cN=(0,2)$ Dirichlet boundary conditions.
One of the main goals of this paper is to elucidate both the structure of chiral algebra and the underlying physics in this setup.
Perturbatively, each boundary supports the affine VOA $V_n(\mathfrak{g})$, where $\mathfrak{g}={\rm Lie }(G)$ and the level is $n=-h^\vee\pm k$ for the left/right boundary, respectively, with $h^\vee$ being the dual Coxeter number.
The full perturbative VOA (that one may compare to the IR one) is given  by $V_{-h^\vee -k}(\mathfrak{g})\otimes V_{-h^\vee +k}(\mathfrak{g})$ extended by a series of bimodules, realized in this case via SUSY Wilson lines (in all representations) stretched between the boundaries (as in Fig. \ref{fig:boundary-conditions}). 
Non-perturbatively, this answer is significantly modified, and we have a complete understanding for the abelian $G$. 
Partial results, conjectures, and challenges of the non-abelian case will be discussed as well.

The IR description of the interval-reduced theory is identified in the companion paper \cite{DL1} as an $\cN=(0,2)$ non-linear sigma-model (NLSM) into the complexification of the gauge group $G_\C \equiv \exp (\mathfrak{g}\otimes \C)$, with the Wess-Zumino (WZ) level $k$.
This is essentially a \emph{non-compact} $\cN=(0,2)$ version of the Wess-Zumino-Witten (WZW) model, which exists for arbitrary Lie group.
\footnote{Note that this differs from the compact $\cN=(0,2)$ WZW models that were previously constructed for even-dimensional target groups, such as $U(1)\times SU(2)$, which all possess complex structure \cite{Spindel:1988nh,Spindel:1988sr,Sevrin:1988ps,Rocek:1991vk,Rocek:1991az,Ang:2018jqi}.} The HT twist in 3d reduces to the purely holomorphic twist in 2d $\cN=(0,2)$ theories, sometimes called the half-twist\footnote{However, more often than not, the term refers to something else, usually in the context of $\cN=(0,2)$ deformations of the $\cN=(2,2)$ theories, see \cite{Katz:2004nn,Sharpe:2006qd,Melnikov:2007xi,Guffin:2008pi,McOrist:2008ji,Melnikov:2009nh,Kreuzer:2010ph,Melnikov:2010sa,Melnikov:2012hk,Guo:2015gha,Gu:2017nye,Gu:2019byn}}.
The perturbative chiral algebra in such models is related to the theory of chiral differential operators (CDO) \cite{gorbounov1999gerbes,gorbounov2001chiral} on the target or, synonymously, curved $\beta\gamma$ systems \cite{nekrasov2005lectures}, as explored in detail in \cite{witten2007}.
The VOA extracted from the sheaf of CDO on $G_\C$ is denoted $\cD_k[G_\C]$ and is the perturbative chiral algebra in the IR sigma-model.
Here $k$ is the modulus of the CDO corresponding to the B-field on $G_\C$ (proportional to the connection on the canonical WZ gerbe on $G_\C$, with $k$ appearing as the proportionality factor).
From the CDO perspective, i.e. in perturbation theory, $k$ does not have to be quantized, and indeed it labels the complex B-field flux in $H^3(G_\C,\C)$.
In our interval model, however, the B-flux is non-generic and related to the CS level in 3D.
For convenience, in this paper the B-flux $k$ is parameterized in such a way that it is precisely equal to the CS level.
It is a mathematical fact that $V_{-h^\vee -k}(\mathfrak{g})\otimes V_{-h^\vee +k}(\mathfrak{g})$ is a sub-VOA of $\cD_k[G_\C]$ \cite{gorbounov2001chiral}, thus the latter is expected to be an extension of the former by bi-modules.
For generic $k$ (which here means $k\in \C \setminus \mathbb{Q}$) such an extension is indeed known to hold \cite{ArkGai,FRENKEL200657,ZHU20081513,Creutzig:2017uxh,CREUTZIG2022108174,Moriwaki:2021epl}:
\begin{equation}
   \label{eq:CDObg} 
\cD_k[G_\C] \cong \bigoplus_{\lambda \in  P_+} V_{\lambda,-h^\vee +k}\left( \mathfrak{g} \right)  \otimes V_{\lambda^*, -h^\vee -k}\left( \mathfrak{g} \right) ,
\end{equation}
where one sums over the dominant weights $\lambda$. 
Physically, the bi-modules in \eqref{eq:CDObg} come from the Wilson lines connecting the boundaries, as stated earlier.

The decomposition \eqref{eq:CDObg} makes perfect sense in perturbation theory, where we are allowed to treat the CS level as a generic complex number.
At physical integer values of $k$, however, the equation \eqref{eq:CDObg} no longer holds, since the structure of $\cD_k[G_\C]$ as a $V_{-h^\vee -k}(\mathfrak{g})\otimes V_{-h^\vee +k}(\mathfrak{g})$ bi-module becomes much more intricate \cite{ZHU20081513} (for non-abelian $G$).
Not only that, but also the nonperturbative corrections are expected to significantly modify it.
Indeed, $\cD_k[G_\C]$ contains universal affine sub-VOAs $V_{-h^\vee -k}(\mathfrak{g})$ and $V_{-h^\vee +k}(\mathfrak{g})$ \cite{gorbounov2001chiral}, not their simple quotients.
Now, the algebra $V_{-h^\vee-k}(\mathfrak{g})$ for $k>0$, (i.e. at the below-critical level,) is already simple, see \cite[Proposition 2.12]{KL1} in the simply-laced case.
However, for $k>h^\vee$, the affine VOA $V_{-h^\vee +k}(\mathfrak{g})$ living on the ``positive'' boundary is not simple: It is very well understood \cite{Kac:1979fz}, has the proper maximal ideal, and the simple quotient denoted $L_{-h^\vee +k}(\mathfrak{g})$.
It was argued in \cite{costello2020boundary,} that the nonperturbative boundary monopoles implement this simple quotient, at least on the ``positive'' boundary, turning $V_{-h^\vee+k}(\mathfrak{g})$ into $L_{-h^\vee +k}(\mathfrak{g})$. 

Thus we expect the exact chiral algebra (for $k>h^\vee$) to contain $V_{-h^\vee -k}(\mathfrak{g})\otimes L_{-h^\vee +k}(\mathfrak{g})$, not $V_{-h^\vee -k}(\mathfrak{g})\otimes V_{-h^\vee +k}(\mathfrak{g})$, which already differs from $\cD_k[G_\C]$. Furthermore, the 3D bulk is expected to admit only finitely many line operators in the IR. Indeed, it is gapped (at least if the interval is long enough) and given by the $G_{k-h^\vee}$ CS for levels\footnote{The IR behavior at arbitrary $k$ is described in \cite{Bashmakov:2018wts}, see also \cite{Affleck:1982as,Ohta:1999iv,Gomis:2017ixy}.} $k>h^\vee$, which only admits finitely many Wilson lines labeled by the integrable representations of $L_{-h^\vee +k}(\mathfrak{g})$ \cite{Witten:1988hf,book:Kac}. We, therefore, expect the exact VOA to be, roughly, $V_{-h^\vee -k}(\mathfrak{g})\otimes L_{-h^\vee +k}(\mathfrak{g})$ extended by such a finite set of bimodules. This appears to be significantly different from $\cD_t[G_\C]$ (e.g., the latter is an infinite extension for generic $t$).

In this paper, we fully address the abelian case, $G=U(1)$, and give partial results on the nonabelian $G$, including the perturbative treatment outlined earlier.
We also speculate on the kind of nonperturbative corrections in the nonabelian $G_\C$ NLSM that are expected to modify $\cD_k[G_\C]$.
We leave the detailed study of such nonperturbative effects for future work.

In the abelian case, $G_\C= \C^*$, so the IR regime is described by the NLSM into $\C^*$, which was analyzed previously in \cite{Dedushenko:2017tdw}.
This theory of course has a $\C^*$-valued $\beta\gamma$ system for its perturbative chiral algebra $\cD_t[\C^*]$.
We think of it either multiplicatively, i.e., $\gamma$ is $\C^*$-valued, or additively, i.e., $\tilde\gamma\sim\tilde\gamma+2\pi i R$, where $R$ is the radius of the compact boson (to be determined by the CS level).
Since the theory is free, the distinction between perturbative and non-perturbative is slightly formal.
In the case of the $\C^*$ model, the natural nonperturbative completion amounts to including twisted sectors into the theory, which correspond to windings around the nontrivial one-cycle in the target.
The twist fields are vortex-like disorder operators, whose 3D origin is, expectedly, the boundary monopoles.
The presence of such sectors extends the $\beta\gamma$ system by the so-called spectrally flowed modules \cite{Ridout:2014,Allen:2020kkt}.
We show that this results in the lattice VOA for the simplest Narain lattice \cite{Narain:1985jj,Narain:1986am}, i.e., $\Z^2\subset \R^2$ with the scalar product whose Gram matrix is $\left(\substack{0\ 1\\1\ 0} \right)$.
In addition to showing this, we also provide the 3D perspective, where each boundary supports a 1D lattice VOA (abelian WZW \cite{Dimofte:2017tpi}), and Wilson lines extend them by the bi-modules.
In total, we find three presentation of the same VOA: (1) as a Narain lattice VOA; (2) as an extension of the $\beta\gamma$ system; (3) as an extension of the abelian ${\rm WZW}_k\otimes {\rm WZW_{-k}}$ by its bimodules.

The nonabelian $G_\C$ NLSM is interacting, which makes the non-perturbative corrections to the perturbative answer $\cD_k[G_\C]$ much more challenging to quantify.
We are certain from the previous discussion, however, that such corrections must be present.
The known instanton effects in 2D $\cN=(0,2)$ theories are vortices that are best understood in the gauged linear sigma model case \cite{Silverstein:1995re,Basu:2003bq,Bertolini:2014dna} (see also \cite{Losev:1999tu} for the $\cN=(2,2)$ case).
In the NLSMs they are captured by the worldsheet wrapping compact holomorphic curves in the target \cite{Dine:1986zy,Dine:1987bq,Beasley:2005iu}, which is notoriously hard to compute except for the simplest models \cite{Tan_2008}.
In our case, however, $G_\C$ does not support any compact holomorphic curves, which basically implies that there must be some \emph{novel} nonperturbative corrections at play.
They correspond to the boundary monopoles that exist in the parent 3D gauge theory, and can be described as new ``noncompact'' vortices in 2D.
Indeed, monopoles are labeled by the cocharacters $e^{ib\varphi}: U(1)\to G$, up to conjugation, which are complexified to $z^b: \C^* \to G_\C$.
This allows to define a natural vortex singularity for the NLSM field $\phi(z,\bar{z})$:
\begin{equation}
    \phi \sim z^b,\quad \text{as } z\to 0.
\end{equation}
Since it is meromorphic, it will define a half-BPS defect.
Dynamics in the presence of such defects is expected to modify the chiral algebra $\cD_t[G_\C]$ appropriately.
We will explore this elsewhere, while here we only focus on the perturbative aspects when $G$ is non-abelian.

Another subtle point is that the presence of such disorder operators, and hence the non-perturbative corrections, in principle, depends on the UV completion.
We will assume that the 3D gauge theory with compact $G$ admits monopoles.\footnote{Which has far-reaching consequences for its dynamics, as, e.g., in Polyakov's argument for confinement in 3D \cite{Polyakov:1976fu}.} 
Without them, the Dirichlet boundary on the ``positive'' end would appear in tension with unitarity \cite{Dimofte:2017tpi,costello2020boundary} (the same issue is not present on the ``negative'' boundary since it supports a non-compact CFT).
Thus the NLSM originating from the gauge theory on the interval must admit the vortex-type disorder operators and the corresponding non-perturbative phenomena.
On the other hand, if one completes the $G_\C$ NLSM in the UV into the LG model as in \cite{DL1}, there is no room for any vortex defects.
In this case the $\cD_t[G_\C]$ is conjectured to give the exact chiral algebra.
Purely from the viewpoint of NLSM, it is conceivable that the knowledge of metric on the target (which is usually ignored in the BPS calculations, but which was computed in \cite{DL1},) allows one to tell apart the cases with and without the nonperturbative corrections.
This question is likewise left for the future.

\paragraph{Motivation via VOA$[M_4]$.} One of the motivations behind this work is to develop a toolkit for computing VOA$[M_4]$ \cite{Dedushenko:2017tdw}.\footnote{A few recent appearances of the interval compactifications in related contexts include \cite{Gaiotto:2017euk,Prochazka:2017qum,Dimofte:2019buf}.} Recall that the VOA$[M_4]$, or more precisely VOA$[M_4, \mathfrak{g}]$, is defined as the chiral algebra in the $\bar{Q}_+$-cohomology of the 2D $\cN=(0,2)$ theory $T[M_4, \mathfrak{g}]$, which is obtained by the twisted compactification of the 6d $(2,0)$ SCFT (of type $\mathfrak{g}$) on the four-manifold $M_4$ \cite{Gadde:2013sca}.

Let us consider a class of four-manifolds $M_4$ admitting a metric with $S^1$ isometry, such that the $S^1$ action is not free. Note that when the $S^1$ action \emph{is} free, the four-manifold is an $S^1$ bundle over some smooth three-manifold $M_3$ (which is a relatively tame class of four-manifolds), and it is conceptually clear how the dimensional reduction simplifies (first reduce on $S^1$ and then on $M_3$). Of course this is still an interesting and nontrivial problem, but our motivation comes from the opposite case, when the $S^1$ action is not free. Some examples of such four-manifolds include but are not limited to: (1) $\Sigma_g\times S^2$, where $\Sigma_g$ is an arbitrary genus-$g$ surface, and the $S^1$ action rotates $S^2$, which is equipped with an $S^1$-invariant ``sausage'' metric; (2) the unique nontrivial sphere bundle over the Riamann surface of genus $g$, denoted as $\Sigma_g\tilde{\times}S^2$; (3) $\cp^2$ with the standard Fubini-Study metric; (4) Hirzebruch surface, given by the connected sum $\cp^2 \# \bar{\cp}^2$, which is actually isomorphic to the nontrivial sphere bundle over a sphere, $S^2\tilde{\times} S^2$; (5) four-sphere $S^4$. In fact, the general class of such four-manifolds is quite well understood, at least in the simply connected compact case. It was shown by Fintushel \cite{Fin1,Fin2} that if a simply connected $M_4$ admits an $S^1$ action (not necessarily an isometry), it must be a connected sum of some number of $S^2\times S^2$, $S^4$, $\cp^2$ and $\bar{\cp}^2$. If $S^1$ is an isometry and the corresponding simply-connected four-manifold is, in addition, non-negatively curved, \cite{SeaYa} proved (see also \cite{HsiaKlei}) that it belongs to the list $\{S^4, \cp^2, S^2\times S^2, \cp^2\# \bar{\cp}^2, \cp^2\# \cp^2\}$. If we only allow codimension-2 fixed loci, then the list is even shorter: $\{S^2\times S^2, \cp^2\# \bar{\cp}^2\}$. Such lists may appear utterly specialized, however, these manifolds present certain interest to us in view of conjectures in \cite{Feigin:2018bkf}, which we aim to check in the future work.

We may consider the twisted compactification on a manifold with $S^1$ isometry in two steps: (1) first reduce the 6D theory along the $S^1$ orbits; (2) then perform reduction along the remaining quotient space $M_4/U(1)$.
The advantage of this procedure is that the first step yields a relatively simple and concrete result -- the maximal 5D SYM (MSYM) with gauge group $G$ (the simply-connected Lie group whose algebra is $\mathfrak{g}$), albeit placed on some curved space with boundaries and, possibly, defects.
For example, for $M_4=\Sigma_g \times S^2$, reducing along the parallels of $S^2$ gives the $\Sigma_g\times I$ geometry, where $I$ is an interval with the principal Nahm pole boundary conditions at both ends \cite{Chacaltana:2012zy}.
Further reduction of the 5D MSYM along $\Sigma_g$ with the topological twist simply gives a 3D $\cN=4$ SYM with $g$ adjoint hypermultiplets.
Thus we end up with the former 3D theory on the interval, with the $(0,4)$ Nahm pole boundary conditions at both ends. In this case, the resulting effective 2D theory in the IR has $(0,4)$ SUSY \cite{Putrov:2015jpa}. For $g=0$ it is expected to be a sigma-model into the framed monopole moduli space \cite{Hitchin:1982gh,Hitchin:1983ay,Donaldson:1984ugy,AtiyahHitchin+2014} (see also \cite{Haghighat:2011xx,Haghighat:2012bm}, where such models appear on magnetic strings in five dimensions, and \cite{Assel:2016lad} for similar 4d sigma models).
Other examples would lead to the interval reductions of various 3D $\cN=2$ gauge theories with matter and CS levels, which would flow to $\cN=(0,2)$ theories in the 2D limit.

We will explore various such examples in the future work \cite{DL3}.
However, it is natural to start with the most basic 3D $\cN=2$ gauge theory, that is the pure SYM, and study the interval VOA in this case.
This is one of the underlying motivations for the current paper.

\paragraph{The rest of this paper is structured as follows.}
In Section \ref{sec:basics} we review the necessary background material.
Then we move on to computing the interval compactification chiral algebra in the 3D $\mathcal{N}=2$ SYM theory with $\cN=(0,2)$ Dirichlet boundaries in Section \ref{sec:3d}.
In Section \ref{sec:2d} we compute the same chiral algebra from the 2D perspective and discuss some issues.
In the abelian case, we end up with three different presentations of the same VOA.
In the nonabelian case, we make general statements when possible, but mostly work with the $G=SU(2)$ example.
Then we finish with some open questions and speculations, and conclude in Section \ref{sec:concl}.

\section{Basics}
\label{sec:basics}
In this section, we set up the conventions and briefly review the background material, including the holomorphic-topological (HT) twist.
In particular, we discuss the 3D $ \mathcal{N}=2$ supersymmetric theory and its twisted content.
We will describe protected sectors, namely, the cohomology of a $ \overline{Q}_{+}$ supercharge in 3D and 2D theories. 
A $ \beta \gamma $ system will be briefly discussed as well.
Also the connection between the cohomology of $ \cN=\left( 0,2 \right)  $ theories and the \v{C}ech cohomology of the $\beta\gamma$ system is expounded upon, as it will be one of the important computational tools later.

\paragraph{Conventions:}
We consider $ \mathbb{R}^{ 2 } \times I $ with Euclidean signature and with coordinates $ x^{ \mu } $ on $  \mathbb{R}^{ 2 } $ and $ t\in[0,L] $ on $ I $.
We choose $ \gamma^{ i }_{\alpha \beta}$ matrices to be the Pauli matrices $ \sigma^{ i } $, and the antisymmetric symbol $\epsilon^{12}=\epsilon_{21}=1$ to lower and raise indices \cite{Intriligator:2013lca}.
\subsection{Basic Supersymmetry}

In Euclidean 3D space spinors lie in a 2-dimensional complex representation of $ SU\left( 2 \right)  $ and the $ \mathcal{N}=2 $ supersymmetry algebra takes the following form:
\begin{equation}
    \begin{split}
        \left\{ Q^{ I }, Q^{ J } \right\} = \delta^{ IJ } \gamma^{ \mu }P_{\mu}.
    \end{split}
\end{equation}
By defining a new combination of supercharges $ Q = Q^{ 1 } + i Q^{ 2 } $ and $ \overline{Q} = Q^{ 1 } - i Q^{ 2 }$, one can obtain the following conventional form of the superalgebra:
\begin{equation}
    \begin{split}
        \left\{ Q,\overline{Q} \right\}  = 2 \gamma^{ \mu } P_{\mu}.
    \end{split}
\end{equation}
Note that the supercharges are not conjugate to each other in Euclidean signature contrary to Minkowski space, where minimal representations are real (Majorana).
This algebra admits a $ U\left( 1 \right) _{R} $-charge, which is an automorphism of this algebra and acts by rotating $ Q $-charges.
Operators also have  a charge $ J_0 $ with respect to Spin(2)$_{E}  $ rotation  parallel to boundaries.
Let us also define the combination $ J \coloneqq \frac{R}{2} - J_0 $, then all charges can be summarized by the following table:
\begin{center}
$ \begin{array}{c|cccc|cc} 
& \bar{Q}_{+} & Q_{+} & \bar{Q}_{-} & Q_{-} & \mathrm{d} z & \mathrm{~d} \bar{z} \\
\hline U(1)_{R} & 1 & -1 & 1 & -1 & 0 & 0 \\
\operatorname{Spin}(2)_{E} & \frac{1}{2} & \frac{1}{2} & -\frac{1}{2} & -\frac{1}{2} & 1 & -1 \\
U(1)_{J} & 0 & -1 & 1 & 0 & -1 & 1
\end{array} $  
\end{center}

In what follows, we will be considering the cohomology of $ Q \defequal \overline{Q}_{+}$.
The $ P_{z} $ is the only $ P_{\mu} $ which is not $ Q $-exact.
It makes our algebra into an algebra with only holomorphic dependence on the coordinates.
We would consider boundary conditions that preserve a  $ \left( 0,2 \right)  $-part of the supersymmetry algebra generated by $ Q_{+} $ and $ \overline{Q}_{+} $. 
We also want to leave $ U\left( 1 \right) _{R} $ unbroken in the bulk and on the boundary.

The only relevant 3D $\cN=2$ multiplet for this paper is a vector multiplet for some gauge group $ G $:
$$
V_{\rm 3D}=\theta \sigma^m \bar{\theta} A_m+i \theta \bar{\theta} \sigma-i \theta^2 \bar{\theta} \bar{\lambda}-i \bar{\theta}^2 \theta \lambda+\frac{1}{2} \theta^2 \bar{\theta}^2 D_{3 d} \quad \text { (WZ gauge) },
$$
where all the fields lie in the Lie algebra $\mathfrak{g} = \operatorname{Lie}(G) $.
It consists of a connection $A_m$, a real scalar $\sigma$, a complex fermion $\lambda_\alpha$, and a real auxiliary field $D_{3 d}$.
We can also define a covariant superfield 
$$
\Sigma_{\rm 3D}=-\frac{i}{2} \epsilon^{\alpha \beta} \bar{D}_\alpha D_\beta V_{3 d}=\sigma-\theta \bar{\lambda}+
\bar{\theta} \lambda+\bar{\theta} \gamma^{ \mu } \theta \epsilon_{\mu \nu \rho} F^{\nu \rho}+i \theta \bar{\theta} D +
\ldots
$$
that satisfies $D_\alpha D^\alpha \Sigma_{\rm 3D}=\bar{D}^\alpha \bar{D}_\alpha \Sigma_{\rm 3D}=0$.

\subsection{Holomorphic-topological twist}
In this section, we review some formulas of the HT-twisted formalism \cite{Aganagic:2017tvx,costello2020boundary}.
The $ Q $-cohomology of operators of the twisted theory and physical theory are the same.
The convenience of this formalism is that some calculations have only a finite number of Feynman diagrams.

The twisted formalism is reviewed nicely in \cite[Section 3.2]{costello2020boundary}. Let us consider a dg-algebra:
    \begin{equation}
        \begin{split}
            \mathbf{\Omega}^{ \bullet } = C^{ \infty }\left( \mathbb{R}^3 \right) \left[ \dd t,\dd\overline{z} \right] ,
        \end{split}
    \end{equation}
    where the multiplication is the multiplication of differential forms.
    One also needs to consider forms with values in the $k$-th power of the canonical line bundle in the $z$-direction:
    \begin{equation}
        \begin{split}
            \Omega^{ \bullet,k } = \Omega^{ \bullet } \otimes K^k  = C^{ \infty }\left( \mathbb{R}^3 \right) \left[ \dd t,\dd\overline{z} \right] \dd z^k.
        \end{split}
    \end{equation}
    
There is cohomological charge $ R $, which is related to the original $R$-charge in the physical theory by adding a ghost charge to it, and a  twisted spin charge $J$.
In the holomorphic-topological twisted 3D $\mathcal{N }=2 $ theory the fields can be organized into the following BV superfields:
\begin{equation}
    \begin{split}
        \mathbf{A} = c + \left( \mathcal{A}_{t}\dd{t} + A_{\overline{z}}\dd{\overline{z}}\right) + B^{ * }_{\overline{z}t}\dd{\overline{z}}\dd{t} \in \Omega^{ \bullet } \otimes \mathfrak{g}\left[ 1 \right],  \\
        \mathbf{B}=\left( B+\cA^{ * }_{\mu}\dd{x^\mu}+ c^{ * }\dd{\overline{z}}\dd{t} \right) \dd{z} \in \Omega^{ \bullet,1 }\otimes  \mathfrak{g}^{ * },
    \end{split}
\end{equation}
where the superfields $ \mathbf{A}  $ and $ \mathbf{B} $ are obtained from the vector multiplet, and we also introduced ghosts. For example, the field
$$
\mathcal{A}:=A_{\bar{z}} \dd{\bar{z}}+\mathcal{A}_t \dd{t}
$$ is just a connection with complexified $\mathcal{A}_t=A_t-i \sigma$ and ordinary $A_{\bar z}$. The field
$B\equiv B_z$ is identified with $\frac{1}{g^2} \overline{\mathcal{F}_{\bar{z} t}}=\frac{1}{g^2} F_{z t}+\ldots$ on shell.
The bracket $[1]$ indicates a shift of cohomological degree by one.
The forms $\dd{t}$ and $\dd{\overline{z}}$ are treated as Grassmann variables, so they anticommute with fermionic fields, and superfields can be regarded as either bosonic or fermionic.

The action of the $ Q $-charge in the twisted formalism can be written as follows:
\begin{equation}
    \label{eq:q_trans} 
        \begin{array}{ll}
        Q \mathbf{A}=F(\mathbf{A}), & Q \mathbf{B}=\mathrm{d}_{\mathbf{A}} \mathbf{B} - \frac{k}{2 \pi} \partial \mathbf{A}. \\
    \end{array}
\end{equation}
Here the differentials are defined as follows:
\begin{equation}
    \dd_{\mathbf{A}} = \dd' - i\mathbf{A},\qquad \dd'=\dd{t} \partial_t + \dd{\bar{z}} \partial_{\bar{z}},\qquad \partial = \dd{z} \partial_z
\end{equation}
and the curvature is
\begin{equation}
    F(\mathbf{A}) = i \dd_{\mathbf{A}}^2 = \dd'\mathbf{A} - i\mathbf{A}^2.
\end{equation}
It will be also useful to keep in mind the following tables of the $ R $ and $ J $ charges of the operators:

\begin{table}[!h]
\begin{center}
\begin{tabular}{c|c|c|c|c}
    & $ c $ &  $\mathcal{A}_{t}$ & ${A}_{  \overline{z}}$ \\
    \hline
     $ R $  & 1 & 0 & 0 \\
    $ J $ & 0 & 0 & -1 \\
\end{tabular}
\qquad
\begin{tabular}{c|c|c|c|c}
     & $ B $&$ \mathcal{A}_{t}^{ \star } $ & $ \mathcal{A}_{\overline{z}}^{ \star } $ & $ c^{ \star } $  \\
     \hline
    $ R $ & 0 & -1 & -1 & -2 \\
    $ J $ & 1 & 1 & 0 & 0 \\
\end{tabular}
\caption{\label{t:charges} The charges of the fields.}
\end{center}
\end{table}
The following charges are assigned to the differential forms:
\begin{table}[h!]
\begin{center}
\begin{tabular}{c|c|c|c}
     & $ \dd{t} $ & $ \dd{\overline{z}} $ & $ \dd{z} $\\
     \hline
    $ R $ & 1 & 1 & 0 \\
    $ J$  & 0 & 1 & -1 \\
\end{tabular}   
\end{center}
\caption{\label{t:chargesRJ} The charges of the differential forms.}
\end{table}
\subsection{\texorpdfstring{$\beta\gamma$}{beta-gamma} System}
In this section we review the $ \beta \gamma $ system, or as it is usually called in mathematical literature, a sheaf of chiral differential operators.
All formulas can be found in  \cite{gorbounov1999gerbes,gorbounov2001chiral,nekrasov2005lectures,witten2007}

Classically, consider a complex manifold $ M $, a map $\gamma: \Sigma\to M$, and a $(1,0)$-form $\beta$ on $\Sigma$ with values in the pullback $\gamma^*(T^* M)$, governed by the following action:
\begin{equation}
    \begin{split}
       \int_{\Sigma} \beta_{i} \overline{\partial }\gamma^{ i },  
    \end{split}
\end{equation}
where $ \gamma^i $ and $ \beta _{i} $ are the holomorphic components of $\gamma$ and $\beta$, respectively.

Quantum mechanically, the situation is more interesting as we want to preserve the OPE's locally.
On each patch we have the usual $ \beta\gamma $ system with the OPE:
    \begin{equation}
        \begin{split}
            \gamma^{ i } \left(z\right)\beta_{j} \left(w\right) \sim \delta ^{ i }_{j}\frac{dw}{z-w},
        \end{split}
    \end{equation}
which in the physics notation yields:
\begin{equation}
    \begin{split}
        \left[ \gamma^{ i }_{n}, \beta_{ j \, k} \right]  = \delta ^{ i }_{j} \delta _{n+ k,0}.
    \end{split}
\end{equation}
The normal ordering prescription for polynomials is defined by the point-splitting procedure and depends on a chosen complex structure.   

To get a global theory, we need to learn how to glue fields on different patches together.
First, let us choose two sets of local coordinates $ \gamma^{ i } $ and $ \widetilde{\gamma}^{ b } $ on some open set.
The gluing is done by a local automorphisms and $ \gamma $ is transformed as in the classical theory.
As we mentioned before, $ \beta $ is transformed classically as $ \beta \mapsto \widetilde{\beta} = f^{ * } \beta $, where $ f  $ is a local holomorphic diffeomorphism.
The quantum version  of this transformation law 
is given by the following general formula \cite{nekrasov2005lectures}:
\begin{equation}
    \label{eq:quan_beta}
\tilde \beta_a=\beta_i \frac{\partial \gamma^i}{\partial \tilde \gamma^a}
 -\underbrace{\frac{1}{2}\left(\partial_j g^{i}_a \partial_i g^{j}_b\right)\frac{\partial \tilde\gamma^b}{\partial  \gamma^k} \partial \gamma^k    }_{\mathclap{} \text{quantum part}} +
\underbrace{\frac{1}{2}\mu _{ ab }\partial \widetilde{\gamma}^{ b }}_{\mathclap{}\text{moduli parameter}},
\end{equation}
where the Jacobian of the transformation is  $ g^{ i }_{a} =  \dfrac{\partial \gamma^{ i }}{\partial \widetilde{\gamma}^{ a }} $.
The ``quantum part'' appears because we want to keep the right OPE on both patches after gluing.
There is an intrinsic ambiguity associated to solving for the OPE equations.
Moreover, the moduli space of the $ \beta \gamma $ system is parametrized by $ \mu $ or, stating it simply,  different ways of gluing our system globally are in one to one correspondence with the possible choices of $ \mu $.
The parameter $ \mu $ takes values in the first \v{C}ech cohomology group with coefficients in the sheaf of closed holomorphic two-forms, i.e. $ \mathcal{H} ^{ 1 }\left( \Omega^{ 2,cl },M\right)  $.

This algebra becomes VOA if $ c_1(M)=0 $. 
The global stress energy tensor is
\begin{equation}
    \begin{split}
        T = - \beta^{ i } \partial  \gamma_{i} - \frac{1}{2}(\log w)'',
    \end{split}
\end{equation}
where $ w $ is the coefficient of the holomorphic top form $ \omega = w d \gamma_1\wedge \ldots\wedge d \gamma_{n} $. 

\subsection{(0,2) Cohomology And \texorpdfstring{$\beta \gamma$}{Beta-Gamma} System}\label{sec:02bgREV}

One of the physics applications of the curved $ \beta \gamma $ system is in the realm of $(0,2)$ theories. 
As discussed in \cite{witten2007} and will be reviewed shortly, the $ \beta \gamma $ system describes the perturbative cohomology of half-twisted $ (0,2) $ theories.

Let us first discuss a general $ (0,2) $ sigma model.
The Lagrangian is constructed locally by introducing a $(1,0)$-form $K=K_{i} \dd{\phi^{i}}$, with complex conjugate $\bar{K}=\bar{K}_{\bar{i}} \dd{\bar{\phi}^{\bar{i}}}$, and setting
$$
I=\int\left|\dd^{2} z\right| \dd{\bar{\theta}^{+}} \dd{\theta^{+}}\left(-\frac{i}{2} K_{i}(\Phi, \bar{\Phi}) \partial_{z} \Phi^{i}+\frac{i}{2} \bar{K}_{\bar{i}}(\Phi, \bar{\Phi}) \partial_{z} \bar{\Phi}^{\bar{i}}\right),
$$
where $\Phi^i$ is a chiral superfield whose bottom component $ \phi^{ i } $ defines a map from a Riemann surface $ \Sigma $ to a target complex manifold $ X $.
The cohomology of the supercharge $ \overline{Q}_{+} $ can be deformed by $ \mathcal{H} = 2i\partial \omega  \in H^{ 1 }\left( M, \Omega^{ 2,cl} \right) $, where $ \omega = \frac{i}{2} \left( \overline{\partial }K - \partial \overline{K} \right)   $. 
Not only that but $ \mathcal{H} $ must be of type $ (2,1) $ to preserve  $ \left( 0,2 \right)  $ supersymmetry.
Note that it is the same class that parametrizes the $ \beta \gamma $ system moduli.

If we set $\alpha^{\bar{i}}=-\sqrt{2} \bar{\psi}_{+}^{\bar{i}}, \rho^i=-i \psi_{+}^i / \sqrt{2}$ and twist the theory then
$ \rho $ is an element of $ \Omega^{ 0,1 }\left( \Sigma \right) \otimes \phi^*\left( TX \right)  $  and $ \alpha $ is from $ \phi^*\left( \overline{TX} \right) $.
Both $ \alpha $ and $ \rho  $ are Grassmann variables.
After the twisting, $ \overline{Q}_{+}$ becomes a worldsheet scalar with the following action on the fields:
\begin{equation}
    \begin{split}
        Q \phi^{ i } = 0,\qquad & Q \overline{\phi}^{ i } = \alpha ^{ i },\\
        Q \rho^{ i }_{\overline{z}} = - \partial _{\overline{z}}\phi^{ i },\qquad &Q \alpha ^{ \overline{i} } = 0,
    \end{split}
\end{equation}
and the action is given by:
\begin{equation}
    \begin{split}
        I = \int \dd^2 z\left( g_{i\overline{j}}\partial _{\overline{z}}\phi^{ i }\partial _{z}\overline{\phi}^{ \overline{j} } + g_{i\overline{j}}\rho ^{ i }_{\overline{z}}\partial _{z}\alpha ^{ \overline{j} } - g_{i \overline{j},\overline{k}} \alpha^{ \overline{k} }\rho^{ i }_{\overline{z}}\partial _{z}\overline{\phi}^{ \overline{j} } \right) + S_T,
    \end{split}
\end{equation}
where $ S_T = - \int \dd^2 z \left( T_{ij}\partial _{\overline{z}}\phi^{ i }\partial _{z}\phi^{ j } - T_{ij,\overline{k}}\alpha ^{ \overline{k} }\rho^{ i }\partial _{z} \phi^{ j } \right) $ and $ \mathcal{H}= \dd T $ should be of the type described above.
We also note that $ T $ is not a 2-form but a 2-gauge field.

Locally, the structure of the $ Q $-cohomology can be understood easily with the help of the $ \beta \gamma $ system.
Consider an open ball $ U_{\alpha} $:
	\begin{equation}
I=\frac{1}{2 \pi} \int_{U_{\alpha}}\left|\dd^{2} z\right| \sum_{i, \bar{j}} \delta_{i, \bar{j}}\left(\partial_{\bar{z}} \phi^{i} \partial_{z} \bar{\phi}^{\bar{j}}+\rho^{i} \partial_{z} \alpha^{\bar{j}}\right).
\end{equation}
All the sections in the cohomology can be written as (for details refer to \cite{witten2007}):
\begin{equation}
F\left(\phi, \partial_{z} \phi, \ldots ; \partial_{z} \bar{\phi}, \partial_{z}^{2} \bar{\phi}, \ldots\right) \in H^{ 0 }\left( {\rm Ops}^{\rm 2d },\overline{Q}^{\rm 2d }_{+} \right).
\end{equation}
If we set $\beta_{i}=\delta_{i \bar{j}} \partial_{z} \bar{\phi}^{\bar{j}}$, which is an operator of dimension $(1,0)$, and $\gamma^{i}=\phi^{i}$ of dimension $(0,0)$, the bosonic part of the  action can be rewritten as: 
\begin{equation}
    I_{U_{\alpha}}^{\beta \gamma}=\frac{1}{2\pi}\int\left|\dd^{2} z\right| \sum_{i} \beta_{i} \partial_{\bar{z}} \gamma^{i}
\end{equation}
and the space of all sections of this theory is
\begin{equation}
F\left(\gamma, \partial_{z} \gamma, \partial_{z}^{2} \gamma, \ldots ; \beta, \partial_{z} \beta, \partial_{z}^{2} \beta \ldots\right). 
\end{equation}
So, locally, the space of sections of the $\beta\gamma$ system and the $Q$-cohomology of the $(0,2)$ sigma model coincide.
Globally, things are a little more complicated and we are required to consider \v{C}ech cohomology to find the operators with all possible $ R $-charges. 
The $ R $-charge in the sigma model description is matched with the cohomological degree:
	\begin{equation}
	\begin{split}
		H^{ \bullet }\left({\rm Ops}^{\rm 2d }, \overline{Q}^{\rm 2d }_{+} \right) 
\cong  H_{\rm \check{C}ech}^{ \bullet }(X,\hat{A}),
	\end{split}
\end{equation}
where $ \hat{A} $ is a sheaf of free $ \beta\gamma $ systems.

\section{3D Perspective}
\label{sec:3d}

In this section, we discuss the $ Q $-cohomology from the 3D $ \mathcal{N}=2 $ point of view. There are a few constructions one could consider. Firstly, the $Q$-cohomology of local operators in the \emph{bulk} is a commutative vertex algebra (VA) $\cV$ intrinsic to the theory \cite{costello2020boundary,Oh:2019mcg}.
Secondly, the $Q$-cohomology of local operators at the \emph{boundary} preserving $(0,2)$ SUSY (explored in the same reference) is, generally speaking, a noncommutative VA.
Thirdly, --- and this is the new structure that we study here, --- one can define the $Q$-cohomology on the \emph{interval}, or the chiral algebra of the interval compactification.
If both the 3D theory and its $(0,2)$ boundary conditions preserve the R-symmetry, this is a vertex operator algebra (VOA), as opposed to just VA, i.e., it necessarily contains the stress energy tensor.
This is obvious since in the IR limit, the theory becomes effectively two-dimensional \cite{DL1}, and the chiral algebra of an R-symmetric 2D $\cN=(0,2)$ theory always has the Virasoro element, as can be seen from the general R-multiplet structure \cite{DedLGChiral}.
In fact, one can also prove this by constructing the $(0,2)$ R-multiplet from the integrated currents directly in 3D \cite{Brunner:2019qyf}.

Intuitively, the interval VOA contains all 3D observables that look like local operators in the 2D limit.
These includes 3D local operators and lines, thus effectively enhancing the $Q$-cohomology of local operators by the line operators stretched between the boundaries.
The line operators can additionally be decorated by local operators in the $Q$-cohomology.
We are allowed to move them to the boundaries, as follows from the properties of $Q$ \cite{costello2020boundary}.
Additionally, the two ends of the line operator can support some other boundary operators that are stuck there and cannot be shifted into the bulk.
Thus the most general configuration in the $Q$-cohomology consists of a line stretched between the boundaries with some local operators sitting at its two endpoints.
This includes the possibility of colliding a boundary operator from the boundary VA mentioned earlier with the endpoint of a line.
On the half-space, the latter implies that lines ending at the boundary engineer modules for the boundary VA \cite{Costello:2018swh}.
In our case, i.e. on the interval, this similarly means that the line operators give bi-modules of the pair of boundary VAs supported at the two ends of the interval.

Examples of line operators that appear here include descendants of the $ Q $-closed local operators integrated over the interval \cite{costello2020boundary}.
Things like Wilson and vortex lines or their generalizations \cite{Dimofte_2020} may appear as well (the Wilson line can be also viewed as a descendent of the ghost field).
We are striving to compute the OPE involving such operators.
In fact, we will compute the exact chiral algebra in the abelian case and the perturbative one in the nonabelian case, that is the OPE of both local and line operators, for gauge theories with the Dirichlet boundary conditions preserving $ \left( 0,2 \right) $ supersymetry.
We will find that the order line operators, i.e. Wilson lines, create representations for the boundary operator algebras.
They are naturally included into the perturbative interval VOA.
The disorder or vortex lines (when allowed), on the other hand, together with the boundary monopoles should be viewed as manifestation of the non-perturbative phenomena.
We claim to fully understand them in the abelian case but only briefly discuss in the nonabelian setting.

Generally speaking, we have chiral algebras on the left and right boundaries denoted by $ \mathcal{V}_{\ell} $ and $ \mathcal{V}_{r} $, respectively. 
There is also the bulk algebra (commutative VA) $  \mathcal{V} $, which includes only local operators. 
Moreover, $\cV$ maps naturally into the left and right algebras via the bulk-boundary maps, allowing to define their tensor product over $\cV$.
There are two maps, which are defined by pushing the local operators from $ \mathcal{V}  $ to the two boundaries:
\begin{equation}
   \begin{split}
       \rho_{\ell, r}:  \mathcal{V} \to  \mathcal{V}_{\ell,r}.
   \end{split}
\end{equation}
Let us first define the algebra that only includes the local operators in 3D:
\begin{equation}
  \begin{split}
      \mathcal{V}_{\ell} \otimes_{\mathcal{V}} \mathcal{V}_{r}.
  \end{split}
\end{equation}
The tensor product over $\cV$ involves the identification of operators that can be obtained from the same operator in the bulk.
The next step is to extend this algebra by modules that correspond to the $ Q $-closed line operators stretched between the boundaries.
We will denote the resulting 3D cohomology as $ H^{ \bullet }({\rm Ops}^{\rm 3d}, Q )  $.

Last but not least, let us note explicitly that in the non-abelian case, we will be mostly discussing the CS level $k>h^\vee$.
The IR physics of a 3D $\cN=2$ YM-CS is known for all values of $k$ \cite{Bashmakov:2018wts}, and for $0<|k|<h^\vee$ it exhibits spontaneous SUSY breaking \cite{Bergman_1999,Ohta:1999iv}, and runaway for $k=0$ \cite{Affleck:1982as}.
What happens on the interval in the range $0\leq|k|< h^\vee$ will be addressed elsewhere, while the $k\geq h^\vee$ case is more straightforward.
Yet, it is interesting enough, as we see in this work.

\subsection{Vector Multiplet}
  
Consider a vector multiplet sandwiched between the Dirichlet boundary conditions. In the twisted formalism this amounts to choosing $ \mathbf{A}\big| = 0 $\cite{costello2020boundary} at both ends.
This, in turn, is equivalent to setting $ \mathrm{c}| = 0 $ and $ \mathcal{A}_{\overline{z}}| = 0 $

The transformation rules for $\mathrm{c}$, $\mathcal{A}$, $B$, and $\cA^*$ in the bulk follow from (\ref{eq:q_trans}):
\begin{equation}
  \label{eq:q_trans_compon}
    \begin{array}{ll}
      Q\, \mathrm{c} \, = - i \mathrm{c}^2,  &Q\, \mathcal{A}\;  = \dd_{\mathcal{A}}\mathrm{c},\\
      Q B=-i[\mathrm{c}, B]-\frac{k}{2 \pi} \partial_{z} \mathrm{c},  &Q\mathcal{A}^* = \dd'B -i \left[ \mathcal{A},B \right] - \frac{k}{2\pi} \partial _{z}\mathcal{A}.
      \end{array}
\end{equation}
where $\dd_\cA=\dd' - i\cA$.

Before diving deeper, we review what is known about the perturbative algebra on the boundary. Recall that the action in the HT twist takes the following form:
\begin{equation}
  \begin{split}
    \int \mathbf{B}F\left( \mathbf{A} \right) + \frac{k}{4\pi} \int \mathbf{A}\partial \mathbf{A}.
  \end{split}
  \label{eq:actionBA} 
\end{equation}
In the twisted formalism, the propagator connects $  \mathbf{A}$ with $ \mathbf{B} $, as follows from the kinetic energy $\mathbf{B}\,\dd'{\mathbf{A}}$. The rest of terms, including the Chern-Simons, are treated as interactions, which induces the bivalent and the trivalent vertices:
\begin{enumerate}
  \item the vertex connecting two $ \mathbf{A} $,
  \item the vertex connecting two $ \mathbf{A} $ and one $ \mathbf{B} $.
\end{enumerate} 
This form of Feynmann rules is very restrictive, and there can only be a finite number of diagrams for a given number of external legs \cite{Gwilliam2019AOE,costello2020boundary}.

For the group $ G $ the  field $ B $ lies in $ \mathfrak{g}^* $. 
The gauge group is broken on the boundary and becomes a global symmetry there. 
There is also a non-trivial boundary anomaly due to the bulk Chern-Simons term and the fermions in the gauge multiplet.
The former contributes $ \pm k $ to the anomaly and the latter contributes $ -h^{ \vee } $.

Thus, we expect to get two boundary affine algebras, one for each boundary global symmetry, with levels dictated by the anomaly.
From (\ref{eq:q_trans_compon}), on the boundary we have $ QB = 0 $.
So, $ B $ is in the cohomology and its OPE with itself was obtained in \cite[section 7.1]{costello2020boundary}: 
\begin{equation}
  \begin{split}
    B_{a} \left(z\right)B_{b} \left(w\right) \sim \frac{\left( -h^{ \vee } \pm k \right) \kappa_{ab}}{(z-w)^{2 }} + \frac{i f_{ab}{}^{ c }}{(z-w)^{ }}B_{c}\left( w \right),
  \end{split}
\end{equation}
where $B_a$ are the components of $B$ in some basis, $\kappa_{ab} $ is the standard bilinear form equal to $\frac1{2h^\vee}$ times the Killing form in that basis, and $ h^{ \vee} $ is the dual Coxeter number.
This expression can be obtained from the charge conservation and anomalies alone.
The $ J $ charge of $ B $ is 1 (see Table \ref{t:charges}). Thus, only $ z $ up to the second power can contribute.
The first term is the anomaly term explained above.

All half-BPS line operators hitting the boundaries are expected to create modules for the boundary algebras, and we will show that it is true perturbatively for the Wilson line momentarily.
A Wilson line can be written as $ \mathcal{P} e^{ \int_{t} A } $.
Observe that it is indeed $Q$-invariant in the usual formalism, or in the twisted formalism by invoking $ Q\left( \mathcal{A} \right)  = \dd_{\mathcal{A}}\mathrm{c}$ and $\mathrm{c}| = 0$.
The kinetic term for $B$ and $\mathcal{A}$ can be written as: 
\begin{equation}
  \begin{split}
    \Tr B\left( \partial _{\overline{z}} \mathcal{A}_{t} - \partial _{t} \mathcal{A}_{\overline{z}} \right).
  \end{split}
\end{equation}
There is a gauge symmetry associated to this term:
\begin{equation}
  \begin{split}
    \cA_{t} &\to \cA_{t} + \partial _{t} \eta, \\
    \cA_{\overline{z}} &\to  \cA_{\overline{z}} + \partial _{\overline{z}}\eta.
  \end{split}
\end{equation}
The propagator for $ B $ and $ \mathcal{ A}_{t} $ with the appropriate gauge fixing \cite{costello2020boundary} is just
\begin{equation}
  \begin{split}
    G(z,\overline{z};t) \propto \frac{\overline{z}}{|x^2|^{ \frac{3}{2} }},\quad \text{where} \quad x^2 = z\bar{z} + t^2,
  \end{split}
\end{equation}
where we do not keep track of a proportionality constant.

Next, we can calculate the OPE of the boundary operator $B(0)$ with the Wilson line segment $W^{(\lambda)}(z)$ in the irreducible representation of $\mathfrak{g}$ labeled by the dominant weight $\lambda$.
Expanding the Wilson line in a Taylor series, one finds that there is only one diagram that can possibly contribute, where a single $A$ from the Wilson line is directly connected to the operator $B$ at the boundary, see Fig. \ref{fig:WaB}.
\begin{figure}[ht]
    \centering
    \def\svgwidth{0.5\columnwidth}
    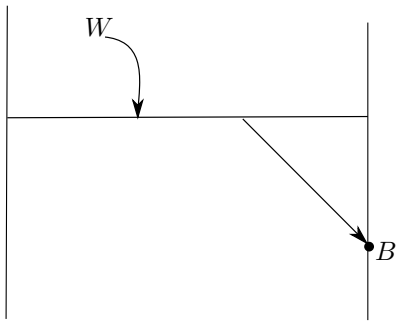

    \caption{The diagram connecting single $\mathcal{A}$ in the Wilson line to the operator $B$ at the boundary.}
    \label{fig:WaB}
\end{figure}
This is proven simply by looking at the two vertices mentioned earlier and realizing that they cannot contribute. 
The diagram evaluates to
\begin{equation}
  \begin{split}
    \int 
    \wick{
        \c1 \cA_{t}^{ a }(t,z,\overline{z}) \c1 B_{ b} 
        \left( 0 \right)\dd{t}
        \propto \delta^a_b\int_0^L \frac{\overline{z}\dd{t}}{\left(t^2+z\overline{z}\right)^{ \frac{3}{2} }} =
        \delta^a_b\frac{L}{z \sqrt{L^2+z \overline{z}} } 
        = \delta^a_b\frac{1}{z} + \ldots,
    }
  \end{split}
\end{equation}
where we keep only the singular term in the $z\to0$ expansion on the right.
This computation already includes corrections to the propagator in the presence of boundaries.
Indeed, such corrections can be accounted for using the method of images.
Since we are interested in the singular term in the OPE, it is enough to only include the first image $\cA_t(-t,z,\bar{z})=\cA_t(t,z,\bar{z})$, as the other ones never get close to $B(0)$.
This simply doubles the contribution of the original insertion $\cA_t(t,z,\bar{z})$.
Combining everything together, this computation shows that the OPE with the Wilson line is
\begin{equation}
  \begin{split}
    B_a\left( z \right) W^{(\lambda)}\left( w \right) \sim \frac{T_a W^{(\lambda)}\left( w \right) }{z-w},
  \end{split}
\end{equation}
where $T_a$ denotes the Lie algebra generator in the same representation $\lambda$ as the Wilson line, and $T_a W^{(\lambda)}(w)$ means the matrix product.
Looking at the Table \ref{t:charges}, we immediately see that this is the only possible OPE as $W^{(\lambda)}$ is not charged. 

More precisely, there are two copies of the affine generators on the interval, denoted $B^\ell_a$ and $B^r_a$ for the left and right boundaries, respectively.
Assuming that the Wilson line performs parallel transport from the right to the left, we find the following OPE's on the interval:
\begin{equation}
  \begin{split}
    B_a^\ell\left( z \right) W^{(\lambda)}\left( w \right) \sim \frac{T_a W^{(\lambda)}\left( w \right) }{z-w},\\
    B_a^r\left( z \right) W^{(\lambda)}\left( w \right) \sim \frac{W^{(\lambda)}\left( w \right) T_a }{z-w}.
  \end{split}
\end{equation}
For completeness, note that the OPE of Wilson line's matrix elements is regular:
\begin{equation}
    W^{(\lambda)}_{ij}(z)W^{(\mu)}_{kl}(w)\sim 0,
\end{equation}
simply because no Feynmann diagram can connect two $\mathcal{A}_t$'s.

Let us pause and contemplate on what we have obtained so far. 
We found that $ B^{\ell,r} $ and $ W^{(\lambda)} $ are elements of the extended cohomology, and we claim that they generate the \emph{perturbative} chiral algebra.
The boundary $ B $'s satisfy the OPE relations of the affine Kac-Moody vertex algebras $ V_{-h^\vee \pm k}(\mathfrak{g}) $.
They also act on $ W^{(\lambda)} $ as on a primary field of the highest weight representation of the affine algebra (i.e., a Weyl module $ V_{\lambda,-h^\vee \pm k}= {\rm Ind}^{ \hat{g} }_{\hat{b}}V_{\lambda} $, where $ V_{\lambda}  $ is a finite-dimensional module for the underlying Lie algebra $\mathfrak{g}$).
Thus, we obtain the following result for the perturbative chiral algebra:
\begin{equation}
  \begin{split}
  \label{eq:PertAlg}
    \cC_k[G_\C] := \bigoplus_{\lambda \in  P_+} V_{\lambda,-h^{ \vee }+k} \otimes V_{\lambda^*, -h^{ \vee }-k},
  \end{split}
\end{equation}
where $ \lambda $ runs over the set $P_+$ of dominant weights, and $\lambda^*=-w(\lambda)$, where $w$ is the longest element of the Weyl group of $G$.
It is not a coincidence that $\cC_k[G_\C]$ looks like $\cD_k[G_\C]$ from the equation \eqref{eq:CDObg} in the Introduction.
This object is well known in the mathematical literature \cite{ArkGai,FRENKEL200657,ZHU20081513,Creutzig:2017uxh,CREUTZIG2022108174,Moriwaki:2021epl},
and away from the rational values of $k$, $\cC_k[G_\C]$ is a simple VOA isomorphic to the VOA $\cD_k[G_\C]$ of chiral differential operators on $G_\C$, with the deformation parameter (perturbative B-field flux) $k\in \C = H^3(G,\C)$. 
At the rational points $k\in\Q$, we can encounter singular vectors, and life is getting much more interesting, e.g., $\cC_k[G_\C]$ and $\cD_k[G_\C]$ are no longer the same \cite{ZHU20081513}.

We can also ask what happens to the stress-energy tensor in our setup. 
We know that it does not exist as a local operator in the bulk chiral algebra \cite[section 2.2]{costello2020boundary}, and generally, the boundary VA does not have to possess a stress-energy tensor as well.
At the same time, we have the current that generates holomorphic translations. 
It acts on the boundary operators as: 
\begin{equation}
  \begin{split}
    \partial_{w} O\left(w \right)  = \int_{HS^{ 2 }} *(T_{z \mu}\dd{x}^{ \mu }) O(w).  
  \end{split}
\end{equation}
There is no boundary part in this expression as we do not introduce any non-trivial degrees of freedom at the boundary.
We can create a line stress-energy operator by stretching the integration surface $HS^2$ to a cylinder in a way that is shown in Fig. \ref{fig:pull-off-the-boundary}.
\begin{figure}[!h]
    \centering
    \def\svgwidth{1\columnwidth}
    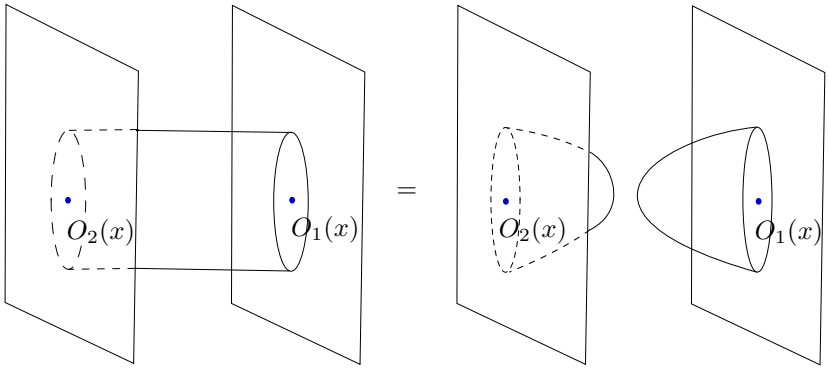

    \caption{Possible codimension one surfaces over which $T_{z\nu}$ is integrated to generate holomorphic translations along the boundary.}
    \label{fig:pull-off-the-boundary}
\end{figure}
This allows to act with the holomorphic translations also on line operators by enclosing them with such a cylinder.
The integration over this tube can be separated into two parts.
The integral over $\dd{t}$ defines the integrated stress-energy tensor,
\begin{equation}
    T_{zz}^{\rm int} = \int_0^L \dd{t} T_{zz},\qquad T_{z\bar{z}}^{\rm int} = \int_0^L \dd{t} T_{z\bar{z}},
\end{equation}
which behaves as a 2D stress tensor generating the holomorphic translations.
The remaining integration over the contour in the boundary plane is reminiscent of the 2D CFT setup.
In fact, it was shown in \cite{Brunner:2019qyf} that such integrated 3D currents (the stress tensor, the R-current and the supercurrents) fit precisely in the 2D $\cN=(0,2)$ R-multiplet.
The presence of this multiplet automatically implies existence of the stress-energy tensor in the cohomology \cite{DedLGChiral}.
In the IR regime, as $ t $ collapses, $T^{\rm int}$ of course becomes the 2D stress tensor, and the 2D $\cN=(0,2)$ arguments are applicable directly.
In either case, we see again that the interval chiral algebra has the stress-energy tensor that follows from the integrated currents.

Outside of critical levels, the boundary algebras are VOAs and have well-defined Sugawara stress tensors.
Physically, we expect that the interval stress-energy tensor becomes the sum of $ T^{ \text{sug} }_\ell $ and $ T^{ \text{sug} }_r $ as an element in the 2D chiral algebra.
It is clear that they act in the same way on the boundary operators.
It indeed turns out to be true as we will argue in the next section using the 2D perspective.

\subsection{The non-perturbative corrections}\label{sec:nonpert}
What are the possible non-perturbative corrections to the above? 
One comes from the boundary monopole operators discussed in \cite{Bullimore:2016nji,Dimofte:2017tpi,costello2020boundary}.
Another possibility is the vortex line connecting the two boundaries.
General gauge vortices discussed in \cite{Kapustin:2012iw,Drukker:2012sr,Hosomichi:2021gxe} are characterized by a singular gauge background $A=b \dd{\varphi}$ close to the vortex locus, where $\varphi$ is an angular coordinate in the plane orthogonal to the line, and $b$ is from the Cartan subalgebra.
This defines a line defect that is local in the plane orthogonal to the vortex, at least away from the boundary, because gauge invariant objects do not feel the gauge holonomy.
This still holds near the Neumann boundary, where the gauge symmetry is unbroken.
However, if such a vortex ends at the Dirichlet boundary, it creates a non-trivial monodromy $e^{2\pi i b}$ for the boundary \emph{global} symmetry.
Hence for generic $b$, it is not a local operator from the 2D boundary point of view as it can be detected far away from the insertion point.
On the interval with the Dirichlet boundaries, such vortices do not lead to local operators in the IR, they become the twist fields that are not included in the VOA.

A less general possibility is a vortex characterised by a nontrivial background $A=b \dd{\varphi}$, yet its monodromy is trivial, $e^{2\pi i b}=1$.
The latter means that $b$ is a co-weight (in fact, a co-character, because the gauge symmetry forces us to consider the Weyl group orbit of $b$).
Thus such a vortex has its magnetic charge labeled by a cocharacter of $G$, or a subgroup $U(1)\hookrightarrow G$ taken up to conjugation.
This is the same as magnetic charges of the monopoles, and we can think of the vortex as an infinitesimal tube of magnetic flux, with the same amount of flux as created by the charge-$b$ magnetic monopole.
This also suggests that monopoles can be located at the endpoints or at the junctions of vortices.

Such vortex lines, however, are not expected to be independent line operators in the IR, at least for a non-zero CS level there.
When $k>h^\vee$, our 3D theory becomes the level $k-h^\vee$ CS theory at large distances.
It has been argued in \cite{Moore:1989yh} that such vortex lines are equivalent to Wilson lines in a CS theory (for the abelian case, see the argument in \cite{Kapustin:2012iw}.)
Thinking of the CS level as a
pairing $K: \Gamma\times \Gamma \to \Z$, where $\Gamma\subset \mathfrak{t}$ is the co-weight lattice of $G$, the representation of the Wilson line is determined precisely by the weight $K(\cdot, b)$.
This is also consistent with the well-known fact (at least in the abelian case \cite{Kapustin:2010hk}) that in the presence of CS level, monopoles develop electric charges, and so Wilson lines can end on them.
In particular, \cite{Kapustin:2010hk} used this to argue that the Wilson lines whose charges differ by a multiple of $k$ are isomorphic, thus showing that the finite spectrum of Wilson lines in a CS theory is a non-perturbative effect manifested via the monopoles.
Note that all the statements we referred to here are about the non-SUSY CS theory, but they extend verbatim to the half-BPS lines in the $\cN=2$ case. 

To apply these observations to the interval theory, it is convenient to assume that the interval is long enough, such that we flow to the CS first and only then to 2D.
At least for $k>h^\vee$ this appears to be a harmless assumption, since SUSY suggests that the BPS sector is not sensitive to the interval length, and the IR physics is also more straightforward in this case \cite{Bashmakov:2018wts}. 

It is therefore natural to conjecture that the non-perturbative effects on the interval with Dirichlet boundaries are captured by the monopoles.
In the bulk they ensure that there are only finitely many inequivalent Wilson lines, and the boundary monopoles modify the boundary VAs.
The details depend strongly on whether $G$ is abelian or non-abelian.

In the non-abelian case, the monopoles at the level $k-h^\vee$ boundary, according to the conjecture in \cite{costello2020boundary}, turn the perturbative affine VOA $V_{k-h^\vee}(\mathfrak{g})$ into its simple quotient $L_{k-h^\vee}(\mathfrak{g})$.
The finitely many bulk Wilson lines correspond to the integrable representations of the latter. 
As for the boundary monopoles at the opposite end, it seems unlikely that they can modify $V_{-k-h^\vee}(\mathfrak{g})$, which is already simple.
The total nonperturbative interval algebra in this case appears to be some modification of the CDO that contains $L_{k-h^\vee}(\mathfrak{g})$ rather than $V_{k-h^\vee}(\mathfrak{g})$.
We do not know its structure yet, and will explore it elsewhere.

The abelian case will be studied in the next sections, where we consider $G=U\left(1\right)$ in detail.
It is a little bit different as there is no abelian WZ term in 2D, and the level is encoded in the periodicity  of the compact boson.
The monopole corrections extend the boundary affine $\mathfrak{u}(1)$ to the lattice VOA, also known as the abelian WZW.
The non-isomorphic Wilson lines correspond to the finite set of modules of the lattice VOA.

\subsection{U(1) }
We restrict $ k $ to lie in $ 2\mathbb{Z} $ and consider the $ k\neq 0 $ case first.
$ B_{\ell},B_{r} $, $ \int \cA_t \dd{t} $ are the possible candidates for elements of the perturbative interval VOA.
The analysis of the OPE of $B$'s still holds, and
$ B $'s on the left and right boundaries commute with each other (have the regular OPE).
So, the full set of OPEs again:
\begin{equation}
\label{eq:BBope}
  \begin{split}
     B_{r}  \left(z\right)B_{r} \left(w\right) &\sim \frac{k}{(z-w)^{2 }},\\
     B_{\ell}  \left(z\right)B_{\ell} \left(w\right) &\sim \frac{-k}{(z-w)^{2 }},\\
     B_{r} (z) B_{\ell}(w) &\sim 0.
  \end{split}
\end{equation}
Surprisingly, there is also one relation connecting the left and right $B$'s to the Wilson line, which follows from the following transformation:
\begin{equation}
  \label{eq:diffBrBlA}
  \begin{split}
   Q\int \mathcal{A}^*_t \dd{t} =  B_{r}-B_{\ell} - \frac{k}{2\pi} \partial_{z}\int\mathcal{A}_t \dd{t}.
  \end{split}
\end{equation}
Thus, we see that the derivatives of $ \int \mathcal{A} $ are not independent operators in the $Q$-cohomology.
We will encounter similar phenomena when we consider singular vectors for affine algebras on the boundaries later.
We can choose any two operators out of $\{B_r,B_\ell,\partial\int \mathcal{A}\}$ as the independent generators, and to be consistent with the previous section, we take $\{B_{\ell,r}\}$.
The stress-energy tensor is also included, but for $ k\neq 0 $ it should be expressed in terms of $ B_{\ell,r} $ as we discussed before. 
In this particular case, $ T = \frac{1}{2k} B_{r}^2 - \frac{1}{2k} B_{\ell}^2$.

We also expect that this perturbative algebra is extended by the boundary monopole operators \cite{Dimofte:2017tpi}.
The boundary monopole $M_p$ is obtained from the usual monopole by cutting it in half and restricting to a half-space in such a way that the integral over the half sphere is
\begin{equation}
  \begin{split}
    \int_{HS^2 } F = 2 \pi p,   \quad p \in \mathbb{Z}.
  \end{split}
\end{equation}
Due to the CS term, the monopole operator $M_p$ develops an electric charge, as was mentioned before.
Hence this operator can only exist by itself on the boundary where its electric charge is global.
Under the global boundary $U(1)$ action by $e^{i \alpha}$ it transforms as:
\begin{equation}
    M_p \to e^{- i p k \alpha} M_p.
\end{equation}
To insert it in the bulk, we need to consider a composite operator with a Wilson line attached to a monopole $e^{ i kp \int_0^t \cA} M_p(t)$ to cancel the anomalous transformation.
In fact, we can pull a boundary monopole $M_p$ off the boundary while extending a Wilson line of charge $kp$ between the monopole and the boundary to respect gauge invariance, as shown in Fig. \ref{fig:Wilson-monopole}.
\begin{figure}[!h]
    \centering
    \def\svgwidth{1\columnwidth}
\begingroup%
  \makeatletter%
  \providecommand\color[2][]{%
    \errmessage{(Inkscape) Color is used for the text in Inkscape, but the package 'color.sty' is not loaded}%
    \renewcommand\color[2][]{}%
  }%
  \providecommand\transparent[1]{%
    \errmessage{(Inkscape) Transparency is used (non-zero) for the text in Inkscape, but the package 'transparent.sty' is not loaded}%
    \renewcommand\transparent[1]{}%
  }%
  \providecommand\rotatebox[2]{#2}%
  \newcommand*\fsize{\dimexpr\f@size pt\relax}%
  \newcommand*\lineheight[1]{\fontsize{\fsize}{#1\fsize}\selectfont}%
  \ifx\svgwidth\undefined%
    \setlength{\unitlength}{220.18819023bp}%
    \ifx\svgscale\undefined%
      \relax%
    \else%
      \setlength{\unitlength}{\unitlength * \real{\svgscale}}%
    \fi%
  \else%
    \setlength{\unitlength}{\svgwidth}%
  \fi%
  \global\let\svgwidth\undefined%
  \global\let\svgscale\undefined%
  \makeatother%
  \begin{picture}(1,0.17882913)%
    \lineheight{1}%
    \setlength\tabcolsep{0pt}%
    \put(0,0){\includegraphics[width=\unitlength,page=1]{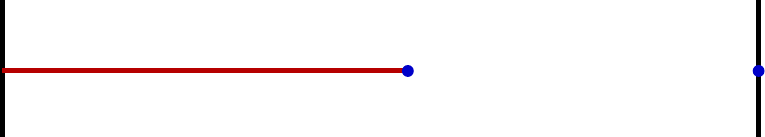}}%
    \put(0.49647652,0.03237282){\color[rgb]{0,0,0}\makebox(0,0)[lt]{\lineheight{1.25}\smash{\begin{tabular}[t]{l}$M_p$\end{tabular}}}}%
    \put(0.21315362,0.10219579){\color[rgb]{0,0,0}\makebox(0,0)[lt]{\lineheight{1.25}\smash{\begin{tabular}[t]{l}$\mathcal{W}_{kp}$\end{tabular}}}}%
    \put(0.70636273,0.07911824){\color[rgb]{0,0,0}\makebox(0,0)[lt]{\lineheight{1.25}\smash{\begin{tabular}[t]{l}$\Longleftrightarrow$\end{tabular}}}}%
    \put(0.93195458,0.032428){\color[rgb]{0,0,0}\makebox(0,0)[lt]{\lineheight{1.25}\smash{\begin{tabular}[t]{l}$M_p$\end{tabular}}}}%
  \end{picture}%
\endgroup%

    \caption{The bulk monopole $M_p$ is connected by a Wilson line of charge $kp$ to the Dirichlet boundary. In the $Q$-cohomology, due to the topological invariance in the $t$ direction, this is equivalent to the boundary monopole $M_p$.}
    \label{fig:Wilson-monopole}
\end{figure}
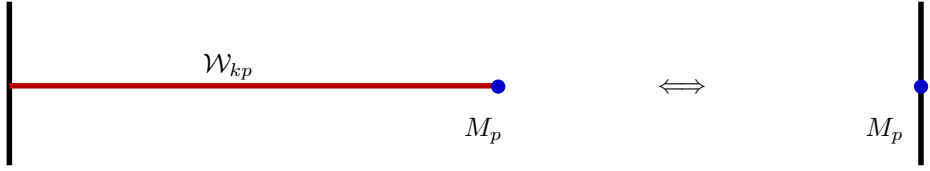
Recall that the twisted theory is topological in the $t$ direction \cite{costello2020boundary}, meaning that the $t$ translations are $Q$-exact.
Thus the length of the Wilson line in Fig. \ref{fig:Wilson-monopole} is irrelevant, and pulling a monopole off the boundary is an identity operation in the cohomology.
Using this observation, we can easily find the OPE of a boundary monopole $M_p$ with the boundary current $B$.
Pulling $M_p$ far away from the boundary, it can no longer contribute to such an OPE.
Essentially, for the purpose of computing OPEs with the boundary operators (in the cohomology), the boundary monopole $M_p$ is equivalent to the Wilson line of charge $kp$ ending at the boundary.
We have already computed the OPE of $B$ with the Wilson line in earlier sections, so the answer follows immediately.
Recalling that there is also a second boundary (and operators on the opposite boundaries have regular OPE), we thus find:
\begin{equation}
\label{eq:BMope}
    \begin{split}
        &B_\ell(z) M_p(0) \sim \frac{k p}{z} M_p(0),\\
        &B_r(z) M_p(0) \sim 0,
    \end{split}
\end{equation}
for the monopole on the left boundary.
A more semi-classical way to derive this is by computing the on-shell value of $B$ in the twisted formalism.
One can easily check that the monopole singularity implies $B\sim \frac{kp}{z}$ (where the factors of $2\pi$ were scaled away).
Similar results hold for monopoles on the right boundary.

In fact, repeating our argument, a monopole on the right boundary is equivalent to the same monopole on the left boundary connected by the Wilson line to the right boundary (and vice versa), see Fig. \ref{fig:monopoles}.
Thus, one can express the right monopoles as $ M_{p}'(t=L) = e^{- i kp \int_0^L \cA } M_{p}(t=0) $, and they are not independent generators.

Via the semiclassical analysis of the monopole operator, one can show that its OPE with the Wilson line is
\begin{equation}
    M_p(z)e^{iq\int_0^L\cA(0,t)} \sim z^{qp} :M_p(z)e^{iq\int_0^L\cA(0,t)}:,
\end{equation}
where $::$ means the normal ordering.
From this equation we see that the normal ordering for $e^{- i kp \int_0^L \cA } M_{p}(t=0)$ is not required: The leading OPE term scales as $z^{p^2}$.
To compute the missing $M_{p_1}(z)M_{p_2}(0)$ OPE, we again use the trick of replacing the left boundary monopole $M_{p_2}$ by the Wilson line attached to the right monopole $e^{ikp_2 \int_0^L\cA_t\dd{t}} M'_{p_2}$,
\begin{equation}
    M_{p_1}(z,0)M_{p_2}(0,0)\sim M_{p_1}(z,0)e^{ikp_2 \int_0^L\cA_t(0,t)\dd{t}}  M'_{p_2}(0,L).
\end{equation}

In this equation, the monopole $M'_{p_2}$ is separated in the $t$ direction from the rest of the operators.
Hence the singular terms only come from the OPE between the Wilson line and the monopole $M_{p_1}$.
This results in the following:
\begin{equation}
    M_{p_1}(z,0)M_{p_2}(0,0)\sim z^{kp_1 p_2} :M_{p_1}(z,0)e^{ikp_2 \int_0^L\cA_t(0,t)\dd{t}}: M'_{p_2}(0,L) .
\end{equation}
It remains to bring $M'_{p_2}$ back to the left boundary, colliding with it along the $t$ direction, which produces no further singularities due to the topological invariance.
Finally, we observe that the magnetic charges are additive under the collision, and conclude:
\begin{equation}
\label{eq:MMope}
    M_{p_1}(z)M_{p_2}(0) \sim z^{k p_1 p_2}:M_{p_1}(z)M_{p_2}(0):  = z^{k p_1 p_2} M_{p_1+p_2}(0)+\dots.
\end{equation}
The full set of strong generators can be chosen to be
\begin{equation}
\label{eq:U1full}
  \begin{split}
    \left< B_{\ell,r},\, e^{i p \int A },\, M_{p}| p \in  \mathbb{Z} \right>.
  \end{split}
\end{equation}
Looking closer at \eqref{eq:BBope}, \eqref{eq:BMope} and \eqref{eq:MMope}, we recognize the rank one lattice VOA setting, with the compact boson of radius $\sqrt{k}$.
We can extend the $ U\left( 1 \right)_{k} $ current $B$ by the monopoles (vertex operators) $M_p$ to the $\Z[\sqrt{k}]$ lattice VOA, also called the $ U(1)_{k} $ WZW.
This can be done individually on the left and right boundaries, giving the abelian WZW$_k$ and WZW$_{-k}$, respectively.
Then the number of Wilson lines is rendered finite, and the algebra is generated as an extension of
\begin{equation}
    {\rm WZW}_k\otimes {\rm WZW}_{-k}
\end{equation}
by the bimodules corresponding to the Wilson lines $e ^{ -i \frac{k-2}{2} \int A }, \ldots, e ^{ i \frac{k}{2}\int A } $.
Indeed, it is consistent with the limit of the large interval.
The theory in this limit reduces to the Chern-Simons at level $k$  in the IR, which only admits a finite number ($k$, to be more precise) of Wilson lines in the bulk, and supports WZW at the boundaries.

Now, let us turn to the $ k=0 $ case.
One can easily see that we no longer have two separate operators $ B_{\ell, r} $, as $ B $ is in the bulk cohomology.
Thus, we can pull it off the boundary and bring it all the way to the opposite one without changing the cohomology class.
Technically, this follows from the relation \eqref{eq:diffBrBlA} involving the descendant line of $B$. 
Monopole operators no longer have electric charge, so they can be freely moved into the bulk, and they commute (have regular OPE) with all local operators.
In other words, monopoles are naturally elements of the bulk VA $\cV$.
The stress-energy tensor is no longer a sum of the two boundary terms but an independent operator.
So, the algebra ceases to be generated as a bimodule of the left and right algebras, and we propose the following set of strong generators:
\begin{equation}
  \begin{split}
    \left<  B,\,T, e ^{ i n \int A }, M_n| n \in \mathbb{Z} \right>.
  \end{split}
\end{equation}

\section{2D Perspective or \texorpdfstring{$\beta\gamma$}{Beta-Gamma} System}
\label{sec:2d}
The IR limit of a 3D $\cN=2$ YM-CS theory on a slab is in general controlled by the dimensionless parameter $\gamma=Le^2_{\rm 3d}$. When $ L \gg \frac{1}{e^{ 2 }_{\rm 3d}} $ or $ \gamma \gg 1 $, the bulk first flows as a 3D theory to the Chern-Simons TFT with $k_{\rm eff}=k-h^\vee$ (for $k\geq h^\vee$). The boundaries flow to some 2D relative CFTs matching the CS boundary anomalies. The interval VOA in this limit was studied in the previous section. We do expect, however, that the answer does not depend on $\gamma$, at least for $k>h^\vee$ (when the SUSY is unbroken).

In the opposite limit $ \gamma \ll 1 $, the system behaves as a 2D sigma model.
It was shown in \cite{DL1} that in this regime, the theory is described by an $\cN=(0,2)$ NLSM into the complexified group $G_\C \approx T^* G$ at level $k$.
The 2D coupling is related to the 3D coupling as $      \frac{L}{e^2_{\rm 3d}} = \frac{1}{e^2_{\rm 2d}}$.
The gauge degrees of freedom are integrated out, except for the complex Wilson line along the interval, which is left as an effective degree of freedom.
Its compact part is valued in $G$, and the vector multiplet scalars lie in the cotangent space.
The Chern-Simons term reduces to a non-trivial $ B $-field, or the Wess-Zumino term, necessary for anomaly matching. 
In this section, we explore our problem from such a 2D viewpoint, as well as study connections between the 3D and 2D setups.

As reviewed in Section \ref{sec:basics}, the perturbative chiral algebra of 2D $\cN=(0,2)$ NLSM into the target $X$ is captured by the $\beta\gamma$ system into $X$ \cite{witten2007,nekrasov2005lectures}. 
More precisely, it is given by the cohomology of a sheaf of $\beta\gamma$ systems on $X$, also called the sheaf of chiral differential operators (CDO) \cite{gorbounov1999gerbes,gorbounov2001chiral,gorbounov2004gerbes}.
It is conveniently computed using the \v{C}ech cohomology, and the resulting vector space with the VA structure on it is denoted $\cD_k[X]$, where $k$ is the B-field.
In our case, the target is a simple complex group $G_\C$, so $k\in \C=H^3(G_\C, \C)$.
This parameter is the well-known modulus of the CDO valued in $H^1(X,\Omega^{2, cl}(X))$, which in general is not quantized, however, in our models $k$ is an integer originating as the CS level in 3D.
Since $G_\C$ has a trivial tangent bundle, $c_1(G_\C)=0$ and $p_1(G_\C)=0$, so the sigma model anomalies \cite{Moore:1984ws} vanish.
As reviewed before, $c_1(G_\C)=0$ implies that $\cD_k[G_\C]$ is a VOA, i.e., it has a well-defined Virasoro element.

Notably, $\cD_k[G_\C]$ containes the affine sub-VOAs $V_{k-h^\vee}(\mathfrak{g})$ and $V_{-k-h^\vee}(\mathfrak{g})$ corresponding to the left and right $G$-actions on $G_\C$ \cite{gorbounov2001chiral}.
These clearly originate as the perturbative baoudnary VOAs in 3D.
Additionally, $\cD_k[G_\C]$ contains holomorphic functions on the group (functions of $\gamma$ in the $\beta\gamma$ language). These originate from the Wilson lines stretched across the interval.
For generic $k\in\C\setminus\Q$, functions on the group together with the pair of affine VOAs $V_{\pm k -h^\vee}(\mathfrak{g})$ generate the whole $\cD_k[G_\C]$. For $k\in\Q$ this is no longer true, and $\cD_k[G_\C]$ as a $V_{k -h^\vee}(\mathfrak{g})\otimes V_{-k -h^\vee}(\mathfrak{g})$--bimodule has a very intricate structure \cite{ZHU20081513}.
Nonetheless, $\cD_k[G_\C]$ remains a simple VOA for all values of $k$.

Below we will find a full non-perturbative answer in the abelian case and comment on what is known in the non-abelian case. By comparing the 3D and 2D answers when possible, we will veryfy that the interpolation between $\gamma\gg 1$ and $\gamma\ll 1$ determines an isomorphism of the respective chiral algebras:
\begin{equation}
   \begin{split}
      H^{ \bullet }\left( {\rm Ops}^{\rm 3d },\overline{Q}_{+} \right) &\xrightarrow{\sim} H^{ \bullet }\left({\rm Ops}^{\rm 2d }, \overline{Q}^{\rm 2d }_{+} \right).
   \end{split}
\end{equation}

\subsection{U(1)}
\label{sec4mon}


For $G=U(1)$, the target space of our model is $ U\left( 1 \right)_{\mathbb{C}}\cong \C^* $, which is a cylinder.
The $\cN=(0,2)$ sigma model into $\C^*$ was considered in \cite{Dedushenko:2017tdw}, and we will make contact with it later.
For now, let us follow the $\beta\gamma$ approach first.
The CDO only have zeroth cohomology in this case, which are the global sections of this sheaf forming the perturbative VOA. Then we will identify its non-perturbative extension.
We work with the even Chern-Simons level $k$ in what follows.

We introduce two cooridnate systems on the cylinder: One is just $\gamma\in\C^*$, and another has $\tilde\gamma$ as a periodic complex boson.
Its periodicity is what encodes the ``level'' in the abelian case, descending from the CS level in 3D:
\begin{equation}
    \tilde\gamma \sim \tilde\gamma + i \pi \sqrt{2k},
\end{equation}
where the conventions are adjusted to match \cite{witten2007}. 
The relation between the two is, naturally,
\begin{equation}
    \gamma = e^{\sqrt{\frac{2}{k}}\tilde\gamma}.
\end{equation}
The quantum transformation between $\beta$ and $\tilde\beta$ is
\begin{equation}
\begin{split}
    \tilde\beta &= \sqrt{\frac{2}{k}}\left(\beta\gamma - \frac12 \gamma^{-1}\partial\gamma \right),\\
    \label{eq:beta_u1_inv}
    \beta &= \sqrt{\frac{k}{2}}\left( \tilde\beta e^{-\sqrt{\frac{2}{k}}\tilde\gamma} - \frac1k e^{-\sqrt{\frac{2}{k}}\tilde\gamma}\partial\tilde\gamma \right).
\end{split}
\end{equation}

Let us define two currents: 
\begin{equation}
   \begin{array}{ll}
      J_{\ell} = \frac{1}{\sqrt{2} } \left( \tilde\beta + \partial  \tilde\gamma \right), & J_{r} = \frac{1}{\sqrt{2} }\left( \tilde\beta - \partial  \tilde\gamma \right),  
   \end{array}
\end{equation}
where $\tilde\gamma \propto \int_{t} \mathcal{A}$ is a (log of a) Wilson line from the 3D perspective, and the difference is $\sqrt{2}\partial\tilde\gamma$, which is proportional to $\partial\int_t \cA$ as in \ref{eq:diffBrBlA}.
One can easily see that these operators commute and have the following OPEs:
\begin{equation}
    \begin{split}
        J_{\ell}(z) J_{\ell}(0) \sim \frac{-1}{z^2}, \quad
        J_{r}(z) J_{r}(0) \sim \frac{1}{z^2}.
    \end{split}
\end{equation}
We observe that they only differ from the 3D  boundary currents $B_{\ell,r}$ by the normalization factor $\sqrt{k}$.
We find such conventions useful for this section.
The stress-energy tensor is $ -\tilde\beta \partial \tilde\gamma $ and the central charge is equal to 2.
The stress-energy tensor does not require any modifications.
The single-valued operator corresponding to the charge $p$ Wilson line is $ W_{p} = \gamma^p$.
The conformal dimensions of $W_p$ can be easily found to be equal to zero.

The operators $\beta$ and $\gamma^p$, $p\in\Z$, are already global sections of the CDO, and they generate the full perturbative VOA, which is most concisely described as \emph{the $\C^*$-valued $\beta\gamma$ system}.
Note that this implies that $\gamma$ can be inverted.
In practice, it is convenient to do so by inverting the zero mode of $\gamma$ only:
\begin{equation}
    \gamma^{-1}(z) = \frac1{\gamma_0 + \Delta\gamma} = \frac{1}{\gamma_0} \sum_{n=1}^{\infty} \left( -1 \right) ^{ n }\left(\gamma_0^{-1}\Delta \gamma \right)^{ n },
\end{equation}
where
\begin{equation}
    \Delta\gamma = \sum_{n\in\Z\setminus 0} \frac{\gamma_n}{z^n}.
\end{equation}

So, what are the non-perturbative effects that we are missing here?
3D physics suggests that they are boundary monopoles.
Imagine inserting an improperly quantized monopole on the boundary with $ \int_{HS^2} F = 2\pi \alpha $ at the origin $z=0$.
We know that $ \gamma $ is related to the connection $ \gamma \propto e^{ i \int_{t}A } $.
Now we can consider moving $\gamma$ around the insertion of this monopole, as in $ \gamma \left( e^{i\phi}z\right)  $, and let us track the phase that is acquired in the process.
It is clear that the Wilson line sweeps the entire magnetic flux of the monopole in the end, and we obtain:
\begin{equation}
\label{eq:monodromy}
   \begin{split}
      \gamma\left( e^{2\pi i} z \right) = e ^{ i \int_{M} F } \gamma \left( z \right) = e ^{ 2 i \pi \alpha } \gamma \left( z \right)  ,
   \end{split}
\end{equation}
where $ M $ is a tube stretched between the two boundaries.
Thus, $  \gamma $ in the presence of the improperly quantized monopole behaves as:
\begin{equation}
\label{eq:shiftmodes}
   \begin{split}
   \gamma = \sum_{n}^{} \frac{\gamma_{n}}{z^{ n-\alpha }}.
   \end{split}
\end{equation}
Now let us go back to the actual monopole that has $\alpha=p\in\Z$.
Equation \eqref{eq:monodromy} shows that $\gamma$ winds $p$ times around the origin as we go once around $z=0$, and we expect the same shift as in \eqref{eq:shiftmodes} with $\alpha=p$.
The actual answer is a little bit trickier, but this gives us a good starting point.

An operator that ``shifts the vacuum'' in the $\beta\gamma$ system  by $p$ units will be called $ M_{p} $.
The module that it creates is known as the spectrally flowed module.\footnote{In superstrings such operators are often called the picture changing operators.}
We can look at the Hilbert space interpretation of these modules.
If we place an $ M_{p} $ operator at the origin, then the mode expansions of $ \beta $ and $ \gamma $  take the following form:
\begin{equation}
   \begin{split}
      \beta  = \sum_{n}^{} \frac{\beta_{n}}{z^{ n+p+1 }},   \quad \gamma = \sum_{n}^{} \frac{\gamma_{n}}{z^{ n-p }}.
   \end{split}
\end{equation}
Here the modes obey the usual commutation relations:
\begin{equation}
    [\beta_n, \gamma_m] = \delta_{n+m,0},
\end{equation}
and the lowest state $|p\rangle$ corresponding to $M_p$ via the state-operator map is defined by
\begin{equation}
    \gamma_{n+1}|p\rangle = \beta_n|p\rangle=0,\quad \text{for } n\geq 0.
\end{equation}
Essentially, we shifted the modes of the vacuum module by $p$ positions and then relabeled them.


The inverse of $ \gamma $ in the presence of $ M_{p} $ is defined in the similar manner:
\begin{equation}
   \begin{split}
      \gamma^{-1}\left( z \right)  = \dfrac{1}{\gamma_0 z^{p}+\Delta \gamma} = \frac{z^{ -p }}{\gamma_0} \sum_{n=0}^{\infty} \left( -1 \right) ^{ n }\left( z^{ p }\gamma_0^{-1}\Delta \gamma \right)^{ n }.
   \end{split}
\end{equation}

Let us compute charges and the conformal dimension of $M_p$.
The currents $J_{\ell},J_{r} $ in the $\gamma$ coordinate system are 
\begin{equation}
\label{eq:ABcurrents}
   \begin{split}
      \begin{pmatrix} 
      J_{\ell}\\
      J_{r}
      \end{pmatrix} 
      = \begin{cases}
         \frac{1}{\sqrt{k} } \left(\beta \gamma+\frac{\left( k-1 \right) }{2}\gamma^{-1}\partial \gamma\right),\\
         \frac{1}{\sqrt{k} }\left(\beta \gamma+\frac{\left( -k-1 \right) }{2}\gamma^{-1}\partial \gamma\right).\\
      \end{cases} 
   \end{split}
\end{equation}

The following product can be computed using the mode expansions:
\begin{equation} \beta\left( z \right) \gamma\left( w \right)  \ket{p} = - \left( \frac{w}{z} \right) ^{ p } \frac{1}{z-w}\ket{p}.
\end{equation}
Subtracting the singularity $ \frac{1}{z-w} $, we are getting:
\begin{equation}
   \begin{split}
       :\beta \gamma: \ket{p} = \frac{p}{z} \ket{p}   .
   \end{split}
\end{equation}
Let us denote $ :\beta \gamma: $ as $ J_0 $ in what follows.
The $ J_0 $ charge of $ M_{p} $ is $ p $, and the charges of $ \beta $ and $ \gamma $ are $ 1 $ and $ -1 $ respectively. 
The difference of $J_\ell$ and $J_r$ is proportional to $\gamma^{-1}\partial\gamma$, which is itself a current measuring the winding number.
In the presence of $M_p$, the expression $ \gamma^{-1}\partial \gamma $ can be computed directly from the definition:
\begin{equation}
   \begin{split}
      \gamma^{-1}\partial \gamma = \frac{p}{z} + \ldots
   \end{split}
\end{equation}
Thus the winding charge is equal to  $ p $ for $ M_{p} $ and 0 for $ \beta $ and $ \gamma $.
The $ U\left( 1 \right)_{\ell,r}  $ charges for $ M_{p}  $ are then computed from \eqref{eq:ABcurrents} and are $\frac{p}{\sqrt{k}}(\frac{1+k}{2},\frac{1-k}{2})$.

\begin{figure}[ht]
    \centering
    \def\svgwidth{0.75\columnwidth}
    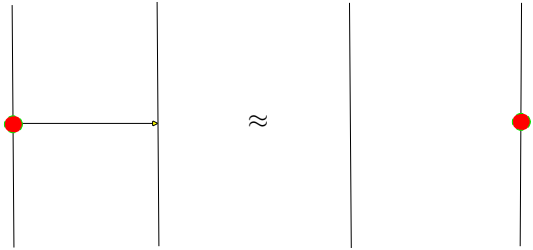

    \caption{Moving a monopole of magnetic charge $p$ from one boundary to another creates a Wilson line of electric charge $kp$.}
    \label{fig:monopoles}
\end{figure}

The stress energy tensor written in terms of $\beta\gamma$ is $ -\beta\partial \gamma + \frac{1}{2} \left(\log \gamma \right) '' $ as the holomorphic top form in this case is $\displaystyle w \propto \frac{\dd \gamma}{\gamma}$.
Moreover, it can be written as $ T_{\beta \gamma} = \frac{1}{2}J_{r}^2-\frac{1}{2}J_{\ell}^2 $, which is indeed what was expected.
Computing the conformal dimensions of $M_p$ requires again the normal ordering prescription:
\begin{equation}
   \begin{split}
      \lim_{z \to w}\left( -\beta\left( z \right) \partial \gamma\left( w \right)  \ket{p} - \frac{1}{\left( z-w \right) ^2}\ket{p}\right) &= \left(\frac{p-p^2}{2w^2}+\frac{\beta_{-1}\gamma_{0}}{w}+\dots\right) \ket{p}, \\
      \frac{1}{2} (\log \gamma)'' \ket{p} &=\left( -  \frac{1}{w^2} \frac{p}{2} + \dots\right) \ket{p},
   \end{split}
\end{equation}
where we again subtracted all singular terms.
It follows that the dimension of the operator $ M_{p} $ is $- \frac{p^2}{2} $.

Now that we understand the operators $M_p$ fairly well, we need to identify precisely those that should be added to the $\beta\gamma$ system to obtain our non-perturbative VOA. They correspond to the boundary monopoles in 3D.
For that we need to find operators that are only charged under the left or right affine currents $J_{\ell,r}$.
If we had $ k=1  $ for a moment, the left monopole is just the flowed module $M_p$, as can be seen from its charges $ \left( J_{\ell},J_{r} \right) = \left( p,0 \right) $. 
The mathematical perspective on such modules was given in \cite{Ridout:2014,Allen:2020kkt}, and in their notations, $M_p$ generates $\sigma^p \cW_0^+$, where $\sigma$ denotes the spectral flow automorphism.

To tackle the general case, we just need to take a composite operator that has the correct charge.
The $ M_{p} $ by itself is charged as $ \frac{p}{\sqrt{k} }  \left(\frac{k+1}{2}, \frac{-k + 1}{2}  \right)  $.
So, if we take $ M_{p}' = (\gamma_0^{ \frac{p(k-1)}{2} } \ket{p})(z)$, we get an operator with the charges $ p (\sqrt{k} ,0) $, as required, and the conformal dimension $ \Delta_{M_p'} = p^2 \frac{k}{2} $.
This $M_p'$ generates the spectrally flowed module $\sigma^p \cW_{\frac{p(k-1)}{2}}$ in the notations of \cite{Allen:2020kkt}, which is just $\sigma^p \cW_0^+$ for $p(k-1)$ even and $\sigma^p \cW_{\frac{1}{2}}$ for $p(k-1)$ odd.\footnote{Since we consider even $k$, this is determined by the parity of $p$ only.}
In order to get monopoles on the other boundary, we need to consider the composite operator with the Wilson line $ W_{k} M'_{1 } $ (Fig. \ref{fig:monopoles}), which has charges $p(0,-\sqrt{k})$ and the same conformal dimension.
Thus we do not need a separate generator for that monopole. 
The generator content of our algebra matches the eqn. (\ref{eq:U1full}).

Let us summarize the result that we have obtained so far for the full nonperturbative VOA.
It is given by the $\C^*$-valued $\beta\gamma$ VOA (i.e., $\gamma$ is invertible,) extended\footnote{In fact, it is slightly redundant to say that $\gamma$ is invertible. The module $\cW_0^+$ of the usual $\beta\gamma$ system coincides with the vacuum module of such a $\C^*$-valued $\beta\gamma$ system. Thus the extension by $\cW_0^+$ automatically inverts $\gamma$.} by its modules $\sigma^{2p} \cW_0^+$ and $\sigma^{2p+1}\cW_{\frac12}$ for all $p\in\Z$, in the notations of \cite{Ridout:2014,Allen:2020kkt}.

One can also easily compute a character over the full chiral algebra:
\begin{equation}
  \begin{split}
  Z = Tr_{ \mathcal{H}}(q^{ L_0-\frac{c}{24} } x^{ J_{\ell} }y ^{ J_{r} }) = \frac{1}{\eta(q)^2} \sum_{n,m \in \mathbb{Z}}^{} q^{nm} x ^{ \frac{n}{\sqrt{k}}+ m \frac{\sqrt{k} }{2} } y ^{ \frac{n}{\sqrt{k}} - m \frac{\sqrt{k} }{2} } ,
  \end{split}
\end{equation}
which is clearly not a meromorphic function and can only be understood as a formal power series.
The obvious non-convergence of the character is expected, as
characters on the boundary with the positive level are convergent when $|q|<1$ \cite{Dimofte:2017tpi}, and on the opposite boundary the convergence is at $|q|>1$.
This trace, of course, can be reinterpreted from the 3D perspective as an index on the interval (see also \cite{Sugiyama:2020uqh}):
\begin{equation}
    Z_{I\times T^2} = Tr_{\mathcal{H}} (-1)^F e^{- 2 \pi R H} \prod e ^{ 2 \pi i J_i z_i}, 
\end{equation}
where $J_i$ are generators of the maximal tori of the boundary symmetries.

The $ k=0 $ situation is different and does not appear to be particularly interesting and well-behaved from the 2D viewpoint, so we skip it.

\subsubsection*{Dual boson}
One can also calculate the same algebra directly in the $\C^*$ sigma model, without going to the $\beta\gamma$-description.
The calculation was first done in \cite{Dedushenko:2017tdw}.
Let us connect it with our formulas for completeness.
Let $ \sigma  $ be the radial coordinate and $ X = X_{L}(z) + X_{R}(\overline{z}) $ be an angular coordinate on $\mathbb{C}^{ * }$.
Then $\tilde\beta$ and $\tilde\gamma$ are related to the free boson as in Sec. \ref{sec:02bgREV}:
\begin{equation}
   \begin{split}
     \partial \tilde\gamma = \partial  \sigma + i \partial X,\\
      \tilde\beta = \partial \sigma - i \partial  X.
   \end{split}
\end{equation}
In the $\gamma$ coordinate system one gets from (\ref{eq:beta_u1_inv}):
\begin{equation}
   \begin{split}
      \gamma &= e ^{ \frac{\sqrt{2} }{\sqrt{k} }\left( \sigma + i (X_{L}+X_R) \right)  },\\
      \beta &= \frac{\sqrt{k} }{{\sqrt{2} } }\left( \partial(\sigma - i X) e ^{ - \frac{\sqrt{2} }{\sqrt{k} }(\sigma + i X) }-
      \frac1k e ^{ - \frac{\sqrt{2} }{\sqrt{k} }\left( \sigma + i X \right)  }\partial( \sigma + i X) \right) \\
      &=-\frac{\sqrt{k}}{\sqrt{2}} \left( (1-\frac1k)\partial\sigma + i  (1+\frac1k) \partial X\right) e ^{ - \frac{\sqrt{2} }{\sqrt{k} }\left( \sigma + i X  \right)  }.
   \end{split}
\end{equation}
Redefine both $ X$ and $\sigma $ by $ \sqrt{2}  $ to match the notations of \cite{Dedushenko:2017tdw}, so we find:
\begin{equation}
   \begin{split}
      \gamma &= e ^{ \frac{ 1 }{\sqrt{k} }\left( \sigma + i (X_{L}+X_R) \right)  },\\
      \beta &=   -\frac{\sqrt{k}}{2} \left( (1-\frac1k)\partial\sigma + i  (1+\frac1k) \partial X\right) e ^{ - \frac{1 }{\sqrt{k} }\left( \sigma + i X  \right)  }.
   \end{split}
\end{equation}
The BPS vertex operators in this description take the form:
\begin{equation}
\label{eq:vertexC}
e^{i k_\ell X_L+k_r\left(\sigma+i X_R\right)},
\end{equation}
with 
\begin{equation}
\left(k_\ell, k_r\right)=\left(\frac{n}{R}+\frac{w R}{2}, \frac{n}{R}-\frac{w R}{2}\right), \quad n, w \in \mathbb{Z}.
\end{equation}
The radius $ R $ is related to the Chern-Simons level as $ R^2=k $.
We can see that these operators form the same lattice as we found in the $\beta \gamma$ system with the special shifted modules included.
In particular, $ \gamma^n $ here is the vertex operator with $\displaystyle k_{\ell}=k_{r} = \frac{n}{R} $.
At the same time the left monopoles have $(k_\ell,k_r)= p( \sqrt{k}  ,0) $, which means $n=wk/2=pk/2$, and the right monopoles have $(k_\ell, k_r)=p (0,-\sqrt{k} ) $, which corresponds to $n=-wk/2=-pk/2$.

\subsection*{No Mercy}
This final presentation allows us to identify the nonperturbative VOA even more explicitly in terms of the known VOAs.
Namely, let us denote the vertex operator representing the $Q$-cohomology class with the momentum and winding charges $(n,w)\in\Z^2$ by $V_{n,w}(z)$.
Then computing the OPE of vertex operators defined in \eqref{eq:vertexC}, we easily find the following:
\begin{equation}
\boxed{
    V_{n_1,w_1}(z) V_{n_2,w_2}(0) \sim z^{n_1 w_2 + n_2 w_1} :V_{n_1,w_1}(z) V_{n_2,w_2}(0):
    },
\end{equation}
which identifies our VOA as a lattice VOA for the smallest Narain lattice \cite{Narain:1985jj,Narain:1986am}, namely $\Z^2\subset \R^2$ with the scalar product:
\begin{equation}
    (n_1, w_1)\circ (n_2,w_2) = n_1 w_2 + n_2 w_1.
\end{equation}
Note that the two $U(1)$ currents can be obtained as $V_{-1,0}\partial V_{1,0}$ and $V_{0,-1}\partial V_{0,1}$:
\begin{equation}
    \begin{split}
        J_1 &= \frac1R \left(i \partial X_L + i\partial X_R + \partial\sigma\right),\\
        J_2 &= \frac{R}{2} \left( i \partial X_L - i\partial X_R - \partial\sigma\right).
    \end{split}
\end{equation}
Also notice a curious fact: While many of the steps in our analysis involved the CS level, the final answer does not depend on it.
This in fact serves as a consistency check for the following reason.
From the $\cN=(0,2)$ point of view, the compact boson radius $\sqrt{k}$ only enters the K\"ahler potential, thus it cannot affect the chiral algebra structure.

Together with the other two results in the earlier sections, we thus find three presentations for the nonperturbative VOA in the abelian case:
\begin{equation}
\label{eq:big_answer}
\boxed{\begin{matrix}\text{Narain lattice VOA}\\ \text{of rank two}\end{matrix}} \cong \boxed{\begin{matrix}\beta\gamma \text{ extended by}\\ \sigma^{2p} \cW_0^+ \text{ and } \sigma^{2p+1}\cW_{\frac12}\end{matrix}} \cong \boxed{\begin{matrix}{\rm WZW}_k\otimes {\rm WZW}_{-k}\\ \text{extended by bimodules}\end{matrix}}
\end{equation}


\subsection{SU(2)}
Let us now turn to a less-trivial example and discuss $G={\rm SU}(2)$, that is $G_\C=$SL $\left( 2,\C \right)  $.
The computation is more involved in this case, as we need to define everything on patches and consider a non-trivial gluing.
We will first find the global theory and discuss the moduli space, and then will turn to the non-trivial modules for boundary VOAs.

 SL$(2,\mathbb{C}) $ can be covered by two patches, the coordinates on which will be denoted as $ \gamma^{ i } $ and $ \widetilde{\gamma}^{ i } $ : 
\begin{equation*}
\begin{pmatrix}
a & b\\
c & d
\end{pmatrix},\ ad-bc=1 \xmapsto{a\ne 0} \begin{pmatrix}
\gamma^1 & \gamma^2\\
 \gamma^3 & {} 
\end{pmatrix}= \begin{pmatrix}
a & b\\
c & {}
\end{pmatrix}\in \mathbb{C}^3\setminus\{\gamma^1=0\},
\end{equation*}
\begin{equation*}
\begin{pmatrix}
a & b\\
c & d
\end{pmatrix},\ ad-bc=1 \xmapsto{b\ne 0} \begin{pmatrix}
\tilde\gamma^1 & \tilde \gamma^2\\
{} & \tilde \gamma^3 
\end{pmatrix}= \begin{pmatrix}
a & b\\
{} & d
\end{pmatrix}\in \mathbb{C}^3\setminus\{\tilde\gamma^2=0\}.
\end{equation*}
Thus, the coordinate transformations have the following form:
\begin{equation}
\label{eq:gamma_transform}
\begin{split}
&\gamma^1=\tilde \gamma^1,\ \gamma^2=\tilde \gamma^2,\ \gamma^3=\frac{\tilde\gamma^1\tilde\gamma^3-1}{\tilde\gamma^2}\ \text{or}\\
&\tilde\gamma^1=\gamma^1,\ \tilde \gamma^2= \gamma^2,\ \tilde\gamma^3=\frac{1+\gamma^2\gamma^3}{\gamma^1}.
\end{split}
\end{equation}
The Jacobian matrix and its inverse can be computed to be
\begin{equation*}
g^i_{\ j}\equiv\frac{\partial \gamma^i}{\partial \tilde \gamma^j}= \begin{pmatrix}
1 & 0 & 0\\
0 & 1 & 0 \\
\frac{1+\gamma^2\gamma^3}{\gamma^1\gamma^2} & -\frac{\gamma^3}{\gamma^2} & \frac{\gamma^1}{\gamma^2}
\end{pmatrix}, \quad \frac{\partial \tilde\gamma^i}{ \partial\gamma^j}= \begin{pmatrix}
1 & 0 & 0\\
0 & 1 & 0 \\
-\frac{1+\gamma^2\gamma^3}{(\gamma^1)^2} & \frac{\gamma^3}{\gamma^1} & \frac{\gamma^2}{\gamma^1}
\end{pmatrix}.  
\end{equation*}
\begin{equation*}
\partial_j g^{i}_a \partial_i g^{j}_b=\partial_3 g^{3}_a \partial_3 g^{3}_b=\begin{pmatrix}
\frac{1}{(\gamma^1)^2} & -\frac{1}{\gamma^1\gamma^2} & 0\\
-\frac{1}{\gamma^1\gamma^2} & \frac{1}{(\gamma^2)^2} & 0 \\
0 & 0 & 0
\end{pmatrix}.   
\end{equation*}

As we mentioned earlier, for each left- and right-invariant vector fields there exists a corresponding VOA sub-algebra \cite{gorbounov2001chiral} that saturates boundary anomalies. 
The boundary with the negative anomaly usually corresponds to a relative CFT with the anti-holomorphic dependence on the coordinates, and it transforms into a chiral algebra with the negative level after passing to the $\overline{Q}$-cohomology: $ k_{coh} = k_{\ell} - k_{r}$ \cite{witten2007}.

Let us find these algebras and the CDO sections explicitly.
Note that there is nothing left except global sections as the manifold is Stein and does not support geometric objects that can form a higher degree cohomology.
So, the only fields that can contribute are in $ H^{ 0 }($SL$(2),\hat{A}) $.

Let us first write out all classical vector fields.
To do this, we will use the following well-known form of basis at the identity point of the group:
\begin{equation}
   e = 
\begin{pmatrix}
  0  & 1 \\
   0 & 0 \\
\end{pmatrix}
   \quad
f = 
\begin{pmatrix} 
   0& 0\\
   1 & 0
\end{pmatrix} 
   \quad
   h = \begin{pmatrix} 
      1 &  0 \\ 
      0 & -1
   \end{pmatrix} 
\end{equation}
and carry it over the whole manifold by $ L_{g}^{ * } V^{ \mu }\partial _{\mu}|_{1} = (gV)^{ \mu } \partial _{\mu}|_{g}  $.

Local sections, corresponding to the left-invariant vector fields, then have the following form in both patches:
\begin{equation}
\begin{array}{ll}

e_\ell=\gamma_1\beta_2 & \tilde e_\ell=\tilde\gamma_1\tilde\beta_2+\tilde\gamma_2^{-1}(\tilde\gamma_1\tilde\gamma_3-1)\tilde\beta_3\\[1.5ex]

f_\ell=\gamma_2\beta_1+\gamma_1^{-1}(1+\gamma_2\gamma_3)\beta_3 {\quad\quad} & \tilde f_\ell=\tilde \gamma_2\tilde\beta_1\\[1.5ex]
h_\ell=\gamma_1\beta_1-\gamma_2\beta_2+\gamma_3\beta_3  & \tilde h_\ell=\tilde\gamma_1\tilde\beta_1-\tilde\gamma_2\tilde\beta_2-\tilde\gamma_3\tilde\beta_3.\\
\end{array}
\end{equation}
Local sections corresponding to the right-invariant vector fields in both patches are
\begin{equation}
\begin{array}{ll}

e_r=\gamma_1\beta_3 &  \tilde e_r=\tilde \gamma_2\tilde\beta_3\\[1.5ex]

f_r=\gamma_3\beta_1+\gamma_1^{-1}(1+\gamma_2\gamma_3)\beta_2 {\quad\quad} & \tilde f_r=\tilde\gamma_2^{-1}(\tilde\gamma_1\tilde\gamma_3-1)\tilde\beta_1+\tilde\gamma_3\tilde\beta_2\\[1.5ex]

h_r=\gamma_1\beta_1+\gamma_2\beta_2-\gamma_3\beta_3  & \tilde h_r=\tilde\gamma_1\tilde\beta_1+\tilde\gamma_2\tilde\beta_2-\tilde\gamma_3\tilde\beta_3.

\end{array}
\end{equation}
The normal ordering for these fields is chosen exactly in the way they are written, $ abc \defequal : a:bc::$, and will be omitted from this point on to unclutter notations.

By using (\ref{eq:quan_beta}), one can easily obtain the following transformation formulas:

\begin{equation}
   \label{eq:quan_beta_sl}
\begin{split}
&\tilde \beta_1=\beta_1+\beta_3\left(\frac{\gamma^3}{\gamma^1}+\frac{1}{\gamma^1\gamma^2}\right)-\frac{1}{2}\left(\frac{1}{(\gamma^1)^2}\partial\gamma^1-\frac{1}{\gamma^1\gamma^2}\partial \gamma^2\right),\\
&\tilde \beta_2=\beta_2-\beta_3\frac{\gamma^3}{\gamma^2}-\frac{1}{2}\left(-\frac{1}{\gamma^1\gamma^2}\partial\gamma^1+\frac{1}{(\gamma^2)^2}\partial \gamma^2\right),\\
&\tilde \beta_3=\beta_3\frac{\gamma^1}{\gamma^2}.
\end{split}
\end{equation}
Note that here we choose to set the moduli parameter $\mu_{ab}$ to zero.
Now we will combine (\ref{eq:gamma_transform}) and (\ref{eq:quan_beta_sl}) to
find the corrected version of these fields.
After either doing a tedious calculation or applying the Mathematica
tool attached, one can obtain the corrected version of the left- and right-invariant vector fields, respectively:
\begin{equation}
\begin{array}{ll}
-H_\ell\equiv h_\ell=\tilde h_\ell& -H_r\equiv h_r+\dfrac{\partial \gamma_1}{\gamma_1}=\tilde h_r+\dfrac{\partial \gamma_2}{\gamma_2} \\[2.5ex]

E_\ell\equiv e_\ell=\tilde e_\ell+\dfrac{1}{2}\partial\left(\dfrac{\tilde \gamma_1}{\tilde\gamma_2}\right) & E_r\equiv e_r=\tilde e_r \\[2.5ex]

F_\ell\equiv f_\ell+\dfrac{1}{2}\partial\left(\dfrac{ \gamma_2}{\gamma_1}\right)=\tilde f_\ell {\quad\quad} &
F_r\equiv f_r+\dfrac{1}{2}\partial\left(\dfrac{\gamma_3}{\gamma_1}\right)+\dfrac{\partial\gamma_3}{\gamma_1}\\&=\tilde f_r+\dfrac{1}{2}\partial\left(\dfrac{\tilde \gamma_3}{\tilde \gamma_2}\right)+\dfrac{\partial\tilde\gamma_3}{\tilde\gamma_2}.
\end{array}
\label{sec:su}
\end{equation}
As one can see, we regrouped terms in the expression, so the sections are actually smooth and well-defined on the whole patch.
For example, the term $ \frac{\partial \gamma_{1}}{\gamma_{1}} $ would have a pole on the second patch, where $ \gamma_{1} $ can be equal to zero. 
Thus, it is only defined smoothly on the first patch.

The OPEs, as expected, constitute the affine Kac-Moody vertex algebra.
For example, the left OPEs are:
\begin{align*}
   H_\ell(z)E_\ell(w)&\sim \frac{2E_\ell(w)}{z-w},  &  H_\ell(z)H_\ell(w)&\sim \frac{-3}{(z-w)^2},\\
   H_\ell(z)F_\ell(w)&\sim \frac{-2F_\ell(w)}{z-w},  & E_\ell(z)F_\ell(w)&\sim \frac{-3/2}{(z-w)^2}+\frac{H_\ell(w)}{z-w}
,\end{align*}
which make it into $V_{-3/2}\left(\mathfrak{sl}(2,\mathbb{C})\right)$.
From the anomaly inflow argument we know that the total level should be $ k_{\ell} + k_{r} = - 2h^{ \vee }=-4$. 
Thus, the right algebra is the affine algebra $ V_{-5/2}\left( \mathfrak{sl}\left( 2, \mathbb{C} \right)  \right)  $, which we could again check by a direct computation.

Operator products between all the left and right global sections are non-singular. The level is defined with respect to the standard bilinear form $(,)=(2h^{\vee})^{-1}(,)_K$, where $(,)_K$ is the Killing form. In the case of $\mathfrak{sl}(2,\mathbb{C})$ it is given by $(h,h)=2,\ (e,f)=(f,e)=1$.

\subsection*{General $ k $}

To find the most general form of this algebra, we need to find the moduli space of this CDO.
By section \ref{sec:basics} we know that it is equivalent to finding $H^{ 1 }\left( \Omega^{ 2,cl },\text{SL}\left( 2,\mathbb{C} \right)  \right) $.
We want to show that 
\begin{equation}
   \label{eq:mu_moduli}
   \begin{split}
      \mu = 2t \frac{\dd{\gamma_1} \wedge \dd{\gamma_2}}{\gamma_1 \gamma_2},\quad t\in\C,
   \end{split}
\end{equation}
is the only generator of that cohomology.
We show it in three steps.
First, we use that $  H^{ 1 }\left(\text{SL}\left( 2,\mathbb{C} \right), \Omega^{ 2,cl } \right)  \to H^{ 3 }_{\textrm dR}\left( \text{SL}\left( 2,\mathbb{C}\right), \mathbb{C}  \right) $ is injective (see Appendix \ref{sec:app}).
Second, it is known that the 3rd de Rham cohomology for simple Lie groups is generated by $ \Tr \left( g^{-1}\dd{g} \right) ^3 $, i.e., $ H^{ 3 }_{\textrm dR}\left( \text{SL}\left( 2,\mathbb{C} \right)  \right) \cong \mathbb{C} $.
Third, the form above is well-defined on $ U_{12} = U_1\cap U_2 \cong \mathbb{C}^{ * }\times \mathbb{C}^{ * }\times \mathbb{C} $ and has a non-trivial period over a non-contractible cycle on the intersection.
Thus, it means that it represents a non-trivial class in $H^1(\text{SL}(2,\C),\Omega^{2,cl})$.
So, we showed that $  H^{ 1 }\left( \Omega^{ 2,cl },\text{SL}\left( 2,\mathbb{C} \right)  \right)  \cong H^{ 3 }_{\textrm dR}\left( \text{SL}\left( 2,\mathbb{C}\right), \mathbb{C}  \right) $ and \ref{eq:mu_moduli} is the non-trivial element. The coefficient $t$ there is thus the only CDO modulus. 

After all these preparations, we can finally redo the calculations with the transformation law shifted by $ \mu $ (\ref{eq:quan_beta}). One gets the following sections:
\begin{equation}
    \begin{array}{ll}
        -H_{\ell} \equiv h_{\ell} + t \dfrac{\gamma_1' 
        }{\gamma_1} = \widetilde{h}_{\ell} - t \dfrac{\gamma_2'}{\gamma_2}   \quad &
        -H_{r} \equiv h_{r} + (1-t) \dfrac{\partial \gamma_1}{\gamma_1} \\
        &=\widetilde{h}_{r} + \left( 1 - t \right) \dfrac{\partial \gamma_2}{\gamma_2}\\
        E_{\ell} \equiv e_{\ell } = t\dfrac{\gamma_1'}{\gamma_2} + \widetilde{e}_{\ell} + \dfrac{1}{2} \partial \left( \dfrac{\gamma_1}{\gamma_2} \right)      \quad &
        E_{r} \equiv e_{r} = \widetilde{e}_{r}\\
        F_{\ell} \equiv f_{\ell} + \dfrac{1}{2} \partial \left( \dfrac{\gamma_2}{\gamma_1} \right) + t\dfrac{\gamma_2'}{\gamma_1} = \widetilde{f}_{\ell}\quad&
        F_r\equiv f_r+\dfrac{1}{2}\partial\left(\dfrac{\gamma_3}{\gamma_1}\right)+(1-t)\dfrac{\partial\gamma_3}{\gamma_1}\\                       &=\tilde f_r+\dfrac{1}{2}\partial\left(\dfrac{\tilde \gamma_3}{\tilde \gamma_2}\right)+(1-t)\dfrac{\partial\tilde\gamma_3}{\tilde\gamma_2}.
    \end{array}
\end{equation}

Effectively,  we observe that introducing the form \eqref{eq:mu_moduli} leads to a shift  of levels $ k_{\ell} \to  k_{\ell} - t $ and $ k_{r} \to  k_{r} +t $.
We will set $ t = \frac{1}{2 } + k $ for convenience, where $k$ now is the actual Chern-Simons level.
The levels of the boundary algebras are then $ -2-k $ and $ -2+k $.
The quantization condition is not necessary in this approach, but is necessary from the 3D perspective.
In this context it means that $ k \in  \mathbb{Z} $.
Thus, the affine algebras of the global left and right $G$-action sections are $ V_{-2 \pm k}\left( \mathfrak{sl}\left( 2,\mathbb{C} \right)  \right)  $.
The explicit form of these sections is one of the key technical results of this chapter.

\subsection*{Module Structure}
To reveal the module structure of $\cD_k[{\rm SL}(2,\C)]$ with respect to $V_{-2 \pm k}\left( \mathfrak{sl}\left( 2,\mathbb{C} \right)  \right)$, let us consider:

\begin{equation}
\begin{array}{ll}
E_\ell(z)\gamma_1(w)\sim 0 & E_\ell(z)(-\gamma_2)(w)\sim \dfrac{\gamma_1(w)}{z-w}\\[2.5ex]
H_\ell(z)\gamma_1(w)\sim\dfrac{\gamma_1(w)}{z-w} & H_\ell(z)(-\gamma_2)(w)\sim \dfrac{\gamma_2(w)}{z-w}\\[2.5ex]
F_\ell(z)\gamma_1(w)\sim \dfrac{-\gamma_2(w)}{z-w} {\quad\quad}& F_\ell(z)(-\gamma_2)(w)\sim 0.
\end{array}
\end{equation}
Thus, as before, $ \gamma $'s generate modules for our boundary algebras and are identified with the Wilson lines in the 3D description.
One finds that we quotient out the singular vector of the underlying $ \mathfrak{sl}_2 $ algebra, i.e. $ (F_{\ell}^{0}){}^2 \gamma_1 = 0 $.

One also finds that the vectors $(\gamma_1)_{0}\ket{0}$, $-(\gamma_2)_{0}\ket{0}$, and all vectors obtained from them by the action of the negative modes of $E_\ell,\ H_\ell,$
and $F_\ell$ span a module over the left current algebra, with the vector $(\gamma_1)_{0}\ket{0}$ being the highest weight vector of weight $1$\footnote{We assumed the mathematical notation, where representations of $\mathfrak{sl}_2$ are labeled by integers, not half-integers.}.
There is an isomorphic module over this subalgebra, ``generated'' by $\gamma_3$ and
$\gamma_{1}^{-1}(1+\gamma_2\gamma_3)$, with the first field giving the highest vector.
One also finds analogous modules over the right
current algebra, ``generated'' by pairs of global sections $\gamma_1, \gamma_3$
and $\gamma_2,\gamma_1^{-1}(1+\gamma_2\gamma_3)$, where again the first field
in each pair defines the highest weight vector of weight $1$. 
Note that these expressions indeed define global sections due to (\ref{eq:gamma_transform}).
All global functions depend only on $(\gamma_1,\gamma_2,\gamma_3,\gamma_1^{-1}(1+\gamma_2\gamma_3))$ and are modules for zero modes of our currents.

Let us put these building blocks together and consider the vector space:
\[
\begin{tikzcd}
    & (\gamma_1)_{0}\ket{0} \arrow[]{dl}{} \arrow[]{dr}{} &  \\
(\gamma_2)_{0}\ket{0} \arrow[]{dr}{} &  & (\gamma_3)_{0}\ket{0} \arrow[]{dl}{} \\
 & \left( \frac{1+\gamma_2\gamma_3}{\gamma_1} \right) _{0}\ket{0} &  
\end{tikzcd}
\]
This is $(1,1)$ representation for $\mathfrak{sl}(2)_\ell\otimes \mathfrak{sl}(2)_r$.
We can act on this vector space by negative modes of $J^a_\ell$ and $J^a_r$. 
The vector $(\gamma_1)_{0}\ket{0}$ is the highest weight vector of weight $(1,1)$ in the representations of the corresponding $ \hat{g}_{\ell} \otimes  \hat{g}_{r}  $ affine Kac-Moody algebras.
In order to obtain other representation one can act with higher powers of $\gamma_1^n$ on the vacuum and this yields  representation $(n,n)$.
So, the answer at the generic point is 
\begin{equation}
\cD_k[\text{SU}(2)_\C] = \bigoplus_{\lambda \in  \mathbb{Z_+}} V_{\lambda,-2 +k}\left( \mathfrak{g} \right)  \otimes V_{\lambda, -2 -k}\left( \mathfrak{g} \right),
\end{equation}
where again $V_{\lambda,-2 \pm k}$ are Weyl modules.

Two points require important clarifications. 
First, what happens with the stress-energy tensor of the $ \beta \gamma $ system $ -\beta_{i}\partial \gamma^{ i } $?
It is guaranteed to exist by \cite{gorbounov1999gerbes} as the canonical bundle on a Lie group is trivial.
The holomorphic top form can be written as $ w = \dfrac{d \gamma_1 d \gamma_2 d \gamma_3}{\gamma_1} $ on the first patch.
Thus, the stress-energy tensor gets corrected to 
\begin{equation}
    \begin{split}
        T\left( z \right)  = -\sum \beta_{i}\partial \gamma_{i} + \frac{1}{2} \left(\log \gamma_1 \right) '',
    \end{split}
\end{equation}
where the correction is a derivative of the coefficient of the holomorphic top form. 
One can show by a direct computation that outside of the critical levels of the boundary algebras,
\begin{equation}
   \begin{split}
      T_{\beta \gamma} = T_{\ell}+ T_{r},
   \end{split}
\end{equation}
where $ T_{\ell,r} = \frac{1}{2(k_{\ell,r}+h^\vee)} \left( ef + fe + \frac{hh}{2}  \right)  $ are the Sugawara stress-energy tensors.
This result was expected from the general discussion in \ref{sec:3d}.

The second point is that the modules that we are considering are reducible for the physical values of $ k \in \mathbb{Z}$.
Not only that, but those singular vectors are singular for both left and right algebras \cite{ZHU20081513}.
To see this, let us set the right algebra level $k_{r} $ to be 0.
It is an obvious limiting case, but it will nevertheless show the important feature that is carried over to other values of $k$.
The simplest singular vector for this module  can be found to be
\begin{equation}
   \begin{split}
      (E_{r})_{-1}\ket{0}.
   \end{split}
\end{equation}
We need to find the form of this vector in terms of the left-invariant fields.
Classically, the vector fields are related by the following change of basis:
\begin{equation}
   \begin{split}
      \begin{pmatrix} 
      e_{r}\\
      f_{r}\\
      h_{r}
      \end{pmatrix}  = \begin{pmatrix} 
      -\gamma_2^2 & \gamma_1^2 & - \gamma_1 \gamma_2\\
      \frac{(1+\gamma_2 \gamma_3)^2}{\gamma^2} & -\gamma_3^2 & \frac{\gamma_3 \left( 1+ \gamma_2 \gamma_3 \right) }{\gamma_1} \\
      2 \frac{\gamma_2 \left( 1+\gamma_2 \gamma_3 \right) }{\gamma_1}                                         &- 2 \gamma_2 \gamma_3 & 1 + 2 \gamma_2 \gamma_3
      \end{pmatrix} 
      \begin{pmatrix} 
      e_{\ell}\\
      f_{\ell}\\
      h_{\ell}
      \end{pmatrix} =
      S\cdot V_{ \ell }
   \end{split}
\end{equation}
Of course, at the quantum level the relation is corrected, and for the $ e $ field the correct answer is found to be
\begin{equation}
   \begin{split}
      E_{r} = V_{\ell}^{ i }S_{1 i} + (-2+k)( \gamma_1 \partial \gamma_2 - \gamma_2 \partial \gamma_1).
   \end{split}
\end{equation}

So, we see that for the special value $ k = 2 $, the correction term disappears, and the vector $ E^{ r }_{-1}  $ can now be obtained both from the left and right algebras. 
It is actually lying inside the $ \gamma_1 ^2 $ representation for the left algebra.
Thus for discrete values of $ k $, different modules start to intersect.
Moreover, now there is no way to obtain this correction term $\omega_f = \gamma_1\partial \gamma_2- \gamma_2 \partial  \gamma_1 $ from a Wilson line by the action of $\hat{\mathfrak{g}}_\ell \otimes \hat{\mathfrak{g}}_r$.
Note that this form is actually dual to the $ f $-vector field $<\omega_f,f> = 1$.
This phenomenon happens for all singular vectors  \cite{ZHU20081513}.


One could ask what happens when we include monopoles. 
We do not have a definitive answer, but as mentioned in the introduction and in Section \ref{sec:nonpert}, we have conjectures as to what the answer might look like. One expects to get some sort of truncation of the CDO that contains simple quotients $L_{-h^\vee\pm k}(\mathfrak{g})$ rather than $V_{-h^\vee\pm k}(\mathfrak{g})$.
We will look into this issue elsewhere.


\subsection{Open Questions}
We have already emphasized many times that determining the non-perturbative modification of $\cD_k[G_\C]$ is an interesting problem.
It requires, perhaps, improved understanding of the non-compact models in 2D, of which our theory is an example.
The usual arguments with quotienting out singular vectors based on the  unitarity of the Hilbert space do not work in such theories.
But we expect that monopoles on the boundary with positive $ k-h^\vee $ are still required, 
as in the opposite limit $ \gamma\gg 1  $ this boundary is  a relative CFT with a normalizable vacuum.
The problems lie on the other boundary which has a  non-compact mode \cite{DL1} that  effectively renders our theory non-compact.

Another intriguing question arises from an alternative UV completion of the $G_\C$ NLSM via the 2D Landau-Ginzburg (LG) model described in \cite{DL1}.
The simplest example is for $ G = SU\left(N  \right)  $.
The UV completion is chosen to be the $\cN=(0,2)$ LG model with the following field content:
\begin{enumerate}
   \item $ M^{ i }_{j} $ are chiral multiplets  valued in complex matrices $ \operatorname{Mat} (N,\mathbb{C}) $; $ \Phi^{ a }_{i} $ is a chiral multiplet that is bifundamental under $ U\left( k \right) \times G $, where $ i,j \in {1,\ldots,N }$ and $ a \in {1,\ldots,(k} = $ anomaly).
   \item Fermi multiplets $ \Gamma $ and $ \Lambda ^{ j }_{b} $. 
   \item Superpotential $ W = \Gamma \left( \operatorname{det} M - 1 \right) + \mu \Lambda_{a}^{ j }M^{ i }_{j} \Phi^{ a }_{ i} $.
\end{enumerate}
This model has the same anomalies as our theory, and the superpotential is engineered in such a way that it flows to the SL$(N,\C)$ NLSM.
It is generally believed that LG models do not carry any non-perturbative physics.
Thus one could hope that the perturbative chiral algebra in this model could provide some useful information.
It is captured by the $\beta\gamma$ systems $(V_i^j,M^i_j)$ (here $V$ is the ``beta'' for $M$), $(R^i_a, \Phi_i^a)$, and the $bc$ systems $(\bar\Gamma, \Gamma)$ and $(\bar\Lambda^a_i, \Lambda_a^i)$.
The chiral algebra is defined in the cohomology of $Q$ that acts according to:
\begin{equation}
\begin{split}
    Q\bar\Gamma &= \det M -1,\qquad\qquad\quad Q \bar\Lambda = M\Phi,\\
    Q V_i^j &= \Gamma \frac{\partial \det M}{\partial M^i_j} + \mu \Lambda^j_a \Phi^a_i,\quad Q R^i_a = \mu \Lambda^j_a M_j^i,
\end{split}
\end{equation}
and by zeros on the rest of fields.
All these $\beta\gamma$ and $bc$ systems are already globally defined, so one simply computes the cohomology of such $Q$.
The answer appears to be just $\cD_k[{\rm SL}(2,\C)]$, which would be interesting to prove.
But more importantly, this, supposedly exact, answer in the LG model is the same as the perturbative VOA we find in the interval theory.
This suggests that the exact non-perturbative physics in these models may depend on the UV completion.

Other open questions involve applications to the VOA$[M_4]$, which requires computing the interval reductions of more complicated gauge theories, and which we will study in the future works.

\section{Conclusion}
\label{sec:concl}
In this paper we considered the chiral algebra of a 3D $\mathcal{N} =2  $ YM on $ \mathbb{R}^2\times \left[ 0,L \right]  $ with the $\cN=(0,2)$ Dirichlet boundary conditions.
The algebra was computed both from the 3D and 2D perspectives.
We analyzed this protected sector using the holomorphic-topological twist of the 3D theory and, among other things, the holomorphic twist of the 2D theory.

From the 3D perspective, the perturbative algebra was found to be an enhancement of two affine vertex algebras living at the boundaries by their bimodules realized via the Wilson lines. The boundary monopoles seem to modify the answer non-perturbatively, on which we proposed some conjectures.

The two-dimensional system after reduction in the right regime is an $\cN= \left( 0,2 \right)  $ NLSM into $ G_{\mathbb{C}} $.
The compactification algebra is the chiral algebra of this  2D model, and the beta-gamma system is the main tool to compute its perturbative approximation. 
The global sections corresponding to the left and right actions of the group on itself were explicitly found for $G={\text{ SU}}(2)$.

In the abelian case, we find that the spectrally flowed modules of the $\beta\gamma$ system are required to get the full result for the algebra.
Combining the latter perspective with the 3D analysis and with the known results on the sigma model into $\C^*$, we obtain three different presentations of the chiral algebra (also called the interval VOA) in \eqref{eq:big_answer}. 
We also saw that the stress-energy tensor is decomposed in terms of the Sugawara stress tensors for the boundary symmetries, both in the abelian and the non-abelian cases.
The answers, when available, fully agree between the 2D and 3D calculations.
Some puzzles and speculations are discussed towards the end and in Section \ref{sec:nonpert}.

\subsection*{Acknowledgements}
We benefited from the useful discussions and/or correspondence with: A.~Abanov, T.~Creutzig, T.~Dimofte, D.~Gaiotto, Z.~Komargodski, I.~Melnikov, N.~Nekrasov, W.~Niu, M.~Ro{\v{c}}ek.
\appendix
\section{De Rham cohomology}
\label{sec:app}
In this appendix we will show that $  H^{1}\left( \text{SL}\left( 2,\mathbb{C} \right), \Omega^{2,\text{cl}} \right)  \to H^3_{\rm dR} \left(\text{SL}(2,\C)\right) $ is injective.
$ H^{1}\left( \text{SL}\left( 2,\mathbb{C} \right), \Omega^{2,\text{cl}} \right)  $ is isomorphic to $\sfrac{\mathcal{Z}_d\left( \Omega^{3,0}\oplus \Omega^{2,1} \right) }{d \Omega^{2,0}}$ \cite{witten2007}.
There is an obvious map from $\sfrac{\mathcal{Z}_d\left( \Omega^{3,0}\oplus \Omega^{2,1} \right) }{d \Omega^{2,0}}$ to the third de Rham cohomology group $H^3_{\mathrm dR}(\text{SL}(2,\mathbb{C}))$ given by $[\alpha]\mapsto [\alpha]$ for any closed 2-form $\alpha \in \Omega^{3,0}\oplus \Omega^{2,1}$. 
\begin{proposition}
For any $[\alpha]\in \sfrac{\mathcal{Z}_d\left( \Omega^{3,0}\oplus \Omega^{2,1} \right) }{d \Omega^{2,0}}$ there exists $\beta\in \cZ_d\left(\Omega^{3,0}\right)$ such that 
\begin{equation}
[\alpha]=[\beta],    
\end{equation}
so there is an isomorphism:
\begin{equation}
\mathcal{A}\equiv\faktor{\mathcal{Z}_d\left( \Omega^{3,0}\oplus \Omega^{2,1} \right) }{d \Omega^{2,0}}\cong \faktor{\mathcal{Z}_d\left( \Omega^{3,0}\right) }{d \Omega^{2,0}},
\end{equation}
where we have made use of a slight abuse of notation, and $d \Omega^{2,0}$ in the last quotient should be understood as $d\Omega^{2,0}\cap \Omega^{3,0}$.
\end{proposition}

\begin{proof}
	A general form from $\mathcal{A}$ has the  form $\alpha+\beta$ for some $ \alpha \in \Omega^{3,0} $ and $ \beta \in  \Omega^{2,1}  $.
	The closeness conditions are
	\begin{equation}
		\begin{split}
			\overline{\partial }\alpha+\partial \beta=0,\\
			\overline{\partial} \beta=0.
		\end{split}
	\end{equation}
	The second condition says that $ \beta \in \mathcal{Z}_{\overline{\partial }}\left( \Omega^{2,1} \right)$, and using the fact that $\text{SL}\left( 2,\mathbb{C} \right) $ is a Stein manifold with all positive degree Dolbeault cohomology groups vanishing $H^{\cdot ,\cdot \ge 1}_{\overline{\partial }}=0$, one gets that the form $\beta$ is in fact exact: $ \beta=\overline{\partial } \gamma $ for some $\gamma\in \Omega^{2,0}$.
  So, shifting $\alpha+\beta$ by $-d\gamma$, we get the desired representative in $\Omega^{3,0}$.\\
	\end{proof}

\begin{proposition}
The map from $\mathcal A$ to $H^3_{\rm dR}$ defined above is injective.
\end{proposition}
\begin{proof}
Suppose we have a closed (3,0)-form $\omega$ that goes under the map to zero in $H^3_{\rm dR}(\text{SL}(2,\mathbb{C}))$, i.e. $\omega=d\alpha$ for some $\alpha\in \Omega^{2,0}\oplus\Omega^{1,1}\oplus\Omega^{0,2}$. Let us represent $\alpha$ as $\alpha^{(2,0)}+\alpha^{(1,1)}+\alpha^{(0,2)}$, where each $\alpha^{(p,q)}\in \Omega^{p,q}$. Thus, 
\begin{equation}
\omega=d\alpha^{(2,0)}+d\alpha^{(1,1)}+d\alpha^{(0,2)},
\end{equation}
and as $\omega \in \Omega^{3,0}$ we find that $\overline{\partial}\alpha^{(0,2)}=0$.
Recalling that $\text{SL}(2,\mathbb{C})$ has trivial non-zero Dolbeault cohomology groups as a Stein manifold, one obtains that $\alpha^{(0,2)}=\overline{\partial}\gamma^{(0,1)}$ for some $\gamma^{(0,1)}\in\Omega^{0,1}$.
It means, however, that $d\alpha^{(0,2)}\equiv\partial\overline\partial\gamma^{(0,1)}+\overline\partial^2\gamma^{(0,1)} = -\overline{\partial }\partial \gamma^{ \left( 0,1 \right)  } = \overline{\partial } \beta^{ \left( 1,1 \right)  }$.
Redefining $ \alpha ^{  \left( 1,1 \right) } $,
\begin{equation}
\omega=d\alpha^{(2,0)}+d\alpha^{(1,1)}.    
\end{equation}
Repeating the same argument with $\alpha^{(1,1)}$, we obtain that $\omega=d\alpha^{(2,0)}$, meaning that it was trivial in $\mathcal A$, which proves the statement.
\end{proof}
Let us show that the only generator of $H^3_{\rm dR}(\text{SL}(2,\mathbb{C}))$ (which is $\operatorname{Tr}(g^{-1}dg)^3$) after mapping to $\mathcal{A}$ indeed corresponds to the closed holomorphic (2,0)-form $\mu$ in $ H^{1}\left( \text{SL}\left( 2,\mathbb{C} \right), \Omega^{2,\text{cl}} \right)$ from the Eq. \eqref{eq:mu_moduli}:     
\begin{equation}
    \mathcal{H}\equiv\Tr\left(g^{-1}dg\right)^3=2 \frac{d\gamma_1\wedge d\gamma_2\wedge d\gamma_3}{\gamma_1}=2\frac{d\tilde\gamma_1\wedge d\tilde\gamma_2\wedge \tilde d\gamma_3}{\tilde\gamma_2}
\end{equation}
Following the general algorithm explicitly constructing the isomorphism between $ H^{1}\left( \text{SL}\left( 2,\mathbb{C} \right), \Omega^{2,\text{cl}} \right)$ and $\mathcal{A}$, we must find two (2,0)-forms $T_1$ and $T_2$ in each of the open sets $U_1$ and $U_2$, such that $\mathcal H=dT_1$ in $U_1$ and $\mathcal H=dT_2$ in $U_2$, and take their difference $T_{12}\equiv T_1-T_2$ in $U_1\cap U_2$ to produce the element of $\mathcal A$:  
\begin{equation}
\frac{d\gamma_1\wedge d\gamma_2\wedge d\gamma_3}{\gamma_1}=d\left( \frac{\gamma_3}{\gamma_1}d\gamma_1\wedge d\gamma_2\right),   
\end{equation}
\begin{equation}
\frac{d\tilde\gamma_1\wedge d\tilde\gamma_2\wedge \tilde d\gamma_3}{\tilde\gamma_2}=d\left( \frac{\tilde\gamma_3}{\tilde\gamma_2}d\tilde\gamma_1\wedge d\tilde\gamma_2\right),   
\end{equation}
\begin{equation}
\frac{\gamma_3}{\gamma_1}d\gamma_1\wedge d\gamma_2-\frac{\tilde\gamma_3}{\tilde\gamma_2}d\tilde\gamma_1\wedge d\tilde\gamma_2=-\frac{d\gamma_1\wedge d\gamma_2}{\gamma_1\gamma_2},    
\end{equation}
which indeed coincides with the Cech cohomology class defined by $\mu$ on $U_1\cap U_2$. 

\printbibliography

\end{document}